\documentclass[10pt]{amsart}
\usepackage{amsbsy,amssymb,amscd,amsfonts,latexsym,amstext,delarray,
amsmath,amsthm,graphicx}

\input xypic

\newtheorem{thm}{Theorem}[section]
\newtheorem{prop}[thm]{Proposition}
\newtheorem{cor}[thm]{Corollary}
\newtheorem{lem}[thm]{Lemma}

\newtheorem{defn}[thm]{Definition}
\newtheorem{rem}[thm]{Remark}
\newtheorem{ex}[thm]{Example}
\newtheorem{conj}[thm]{Conjecture}

\numberwithin{equation}{section}

\def\bG{{\mathbb G}}

\def\bL{{\mathbb L}}

\def\bT{{\mathbb T}}

\def\A{{\mathbb A}}

\def\N{{\mathbb N}}
\renewcommand{\P}{{\mathbb P}}
\def\Q{{\mathbb Q}}
\def\Z{{\mathbb Z}}
\def\R{{\mathbb R}}

\def\one{{\bf 1}}

\def\cC{{\mathcal C}}

\def\cG{{\mathcal G}}
\def\cH{{\mathcal H}}
\def\cI{{\mathcal I}}

\def\cL{{\mathcal L}}
\def\cM{{\mathcal M}}

\def\cS{{\mathcal S}}

\def\cU{{\mathcal U}}
\def\cV{{\mathcal V}}

\newcommand{\ie}{{\it i.e.\/}\ }
\newcommand{\eg}{{\it e.g.\/}\ }
\newcommand{\cf}{{\it cf.\/}\ }
\def\text{\hbox}

\def\End{{\rm End}}

\def\GL{{\rm GL}}

\def\Spec{{\rm Spec}}

\title[Banana motives]{Feynman motives of banana graphs}
\author[Aluffi]{Paolo Aluffi}
\author[Marcolli]{Matilde Marcolli}
\address{Department of Mathematics \\ Florida State University \\
Tallahassee, FL 32306, USA}
\email{aluffi\@@math.fsu.edu}
\email{marcolli\@@math.fsu.edu}
\address{Max--Planck Institut f\"ur Mathematik  \\
Vivatsgasse 7 \\
Bonn, D 53111, Germany}
\email{aluffi\@@mpim-bonn.mpg.de}
\email{marcolli\@@mpim-bonn.mpg.de}

\begin{document}

\maketitle

\begin{abstract}
We consider the infinite family of Feynman graphs known
as the ``banana graphs'' and compute explicitly the 
classes of the corresponding graph hypersurfaces in the
Grothendieck ring of varieties as well as their 
Chern--Schwartz--MacPherson classes, using the classical 
Cremona transformation and the dual graph, and a blowup
formula for characteristic classes. We outline the
interesting similarities between these operations and we 
give formulae for cones obtained by simple operations 
on graphs. We formulate a positivity conjecture for 
characteristic classes of graph hypersurfaces and
discuss briefly the effect of passing to noncommutative
spacetime.
\end{abstract}

\section{Introduction}\label{IntroSec}

Since the extensive study of \cite{BroKr} revealed the systematic
appearance of multiple zeta values as the result of Feynman diagram
computations in perturbative quantum field theory, the question of
finding a direct relation between Feynman diagrams and periods of
motives has become a rich field of investigation. The formulation
of Feynman integrals that seems most suitable for an algebro-geometric
approach is the one involving Schwinger and Feynman parameters, as
in that form the integral acquires directly an interpretation as a
period of an algebraic variety, namely the complement of a
hypersurface in a projective space constructed out of the
combinatorial information of a graph. These graph hypersurfaces
and the corresponding periods have been investigated in the
algebro-geometric perspective in the recent work of
Bloch--Esnault--Kreimer (\cite{Blo}, \cite{BEK}) and more recently,
from the point of view of Hodge theory, in \cite{BloKr} and
\cite{Mar}.  In particular, the question of whether only motives of 
mixed Tate type would arise in the quantum field theory context is 
still unsolved. Despite the general result of \cite{BeBro}, which
shows that the graph hypersurfaces are general enough from the
motivic point of view to generate the Grothendieck ring of varieties,
the particular results of \cite{BroKr} and \cite{BEK} point to
the fact that, even though the varieties themselves are very general,
the part of the cohomology that supports the period of interest to
quantum field theory might still be of the mixed Tate form.

One complication involved in the algebro-geometric computations with
graph hypersurfaces is the fact that these are typically singular,
with a singular locus of small codimension. It becomes then an
interesting question in itself to estimate how singular the graph
hypersurfaces are, across certain families of Feynman graphs (the half
open ladder graphs, the wheels with spokes, the banana graphs etc.).
Since the main goal is to describe what happens at the motivic level,
one wants to have invariants that detect how singular the hypersurface
is and that are also somehow adapted to its decomposition in the
Grothendieck ring of motives. 
In this paper we concentrate on a particular example and illustrate
some general methods for computing such invariants based on the theory
of characteristic classes of singular varieties. 

\medskip

Part of the purpose of the present paper is to familiarize
physicists working in perturbative quantum field theory with some
techniques of algebraic geometry that are useful in the analysis of
graph hypersurfaces. Thus, we try as mush as possible to spell out
everything in detail and recall the necessary background.

\medskip

In \S \ref{IntroSec},
we begin by recalling the general form of the
parametric Feynman integrals for a scalar field
theory and the construction of the associated
projective graph hypersurface. We recall the
relation between the graph hypersurface of a
planar graph and that of the dual graph via 
the standard Cremona transformation.
We then present the specific example of the
infinite family of ``banana graphs''. We formulate 
a positivity conjecture for the characteristic classes
of graph hypersurfaces.

For the convenience of the reader, we recall in \S \ref{GrCSMdefSec}
some general facts and results, both about the Grothendieck ring
of varieties and motives, and about the theory of
characteristic classes of singular algebraic varieties. 
We outline the similarities and differences between
these constructions.

In \S \ref{MotSect} we give the explicit computation
of the classes in the Grothendieck ring of the
hypersurfaces of the banana graphs. We conclude with
a general remark on the relation between the class of
the hypersurface of a planar graph and that of a dual
graph. 

In \S \ref{CSMsect} we obtain an explicit formula for
the Chern--Schwartz--MacPherson classes of the
hypersurfaces of the banana graphs. We first prove a
general pullback formula for these classes, which is
necessary in order to compute the contribution to the
CSM class of the complement of the algebraic simplex
in the graph hypersurface. The formula is then obtained
by assembling the contribution of the intersection with
the algebraic simplex and of its complement via inclusion--exclusion,
as in the case of the classes in the Grothendieck ring.

We give then, in \S \ref{ConesSect}, a formula for the
CSM classes of cones on hypersurfaces and use them to
obtain formulae for graph hypersurfaces obtained from
known one by simple operations on the graphs, such as
doubling or splitting an edge, and attaching single-edge 
loops or trees to vertices. 

Finally, in \S \ref{NCQFTsect}, we look at the deformations of
ordinary $\phi^4$ theory to a noncommutative spacetime given
by a Moyal space. We look at the ribbon graphs that correspond
to the original banana graphs in this noncommutative quantum
field theory. We explain the relation between the graph
hypersurfaces of the noncommutative theory and of the original
commutative one. We show by an explicit computation of CSM
classes that in noncommutative QFT the positivity conjecture
fails for non-planar ribbon graphs.

\medskip

{\bf Acknowledgment.} The first author is partially supported by NSA
grant H98230-07-1-0024. The second author is partially supported by 
NSF grant DMS-0651925. We thank the Max--Planck--Institute and Florida
State University, where part of this work was done. We also thank 
Abhijnan Rej for exchanges of numerical computations of CSM classes 
of graph hypersurfaces. 

\subsection{Parametric Feynman integrals}

We briefly recall some well known facts (\cf \S 6-2-3 of \cite{ItZu}, 
\S 18 of \cite{BjDr}, and \S 6 of \cite{Naka}) about the parametric 
form of Feynman integrals. 

Given a scalar field theory with Lagrangian written in Euclidean signature as
\begin{equation}\label{Lagr}
 \cL(\phi)= \frac 12 (\partial \phi)^2 + \frac{m^2}{2} \phi^2 +
\cL_{int}(\phi),
\end{equation}
where the interaction part is a polynomial function of $\phi$, a
one-particle-irreducible (1PI) Feynman graph of the theory is
a connected graph $\Gamma$ which cannot be disconnected by removing a
single edge, and with the following properties. All vertices in
$V(\Gamma)$ have valence equal to the degree of one of the monomials
in the Lagrangian. The set of edges
$E(\Gamma)=E_{int}(\Gamma)\cup E_{ext}(\Gamma)$ consists of internal 
edges having two end vertices and external ones having only one
vertex. A Feynman graph without external edges is called a vacuum
bubble. 

In perturbative quantum field theory, the Feynman integrals associated
to the loop number expansion of the effective action for a scalar
field theory are labeled by the 1PI Feynman graphs of the theory, each
contributing a corresponding integral of the form
\begin{equation}\label{FeynInt}
U(\Gamma,p)=   \frac{\Gamma(n-D\ell/2)}{(4\pi)^{\ell D/2}}  
\int_{[0,1]^n} \frac{\delta(1-\sum_i t_i)}
{\Psi_\Gamma(t)^{D/2}V_\Gamma(t,p)^{n- D\ell/2}} \, dt_1 \cdots dt_n. 
\end{equation}
Here $n=\# E_{int}(\Gamma)$ is the number of internal edges of the
graph $\Gamma$, $D\in \N$ is the spacetime dimension in which the
scalar field theory is considered, and $\ell =b_1(\Gamma)$ is the
number of loops in the graph, \ie the rank of $H_1(\Gamma,\Z)$. 
The function $\Psi_\Gamma$ is a polynomial of degree
$\ell=b_1(\Gamma)$. It is given by the Kirchhoff polynomial
\begin{equation}\label{PsiGamma}
\Psi_\Gamma(t)= \sum_{T\subset \Gamma} \prod_{e \notin E(T)} t_e,
\end{equation}
where the sum is over all the spanning trees $T$ of $\Gamma$. The
function $V_\Gamma(t,p)$ is a rational function of the form
\begin{equation}\label{ratioVGamma}
V_\Gamma(t,p) = \frac{P_\Gamma(t,p)}{\Psi_\Gamma(t)},
\end{equation}
where $P_\Gamma$ is a homogeneous polynomial of degree
$\ell+1=b_1(\Gamma)+1$ of the form 
\begin{equation}\label{PGammapt}
P_\Gamma(p,t) = \sum_{C\subset \Gamma} s_C \prod_{e\in C} t_e, 
\end{equation}
Here the sum is over the cut-sets $C\subset \Gamma$, \ie the
collections of $b_1(\Gamma)+1$ edges that divide the graph $\Gamma$ in
exactly two connected components $\Gamma_1\cup \Gamma_2$. 
The coefficient $s_C$ is a function of the external momenta attached 
to the vertices in either one of the two components
\begin{equation}\label{sCcoeff}
s_C = \left(\sum_{v\in V(\Gamma_1)} P_v\right)^2 = \left(\sum_{v\in
V(\Gamma_2)} P_v\right)^2, 
\end{equation}
where the $P_v$ are defined as
\begin{equation}\label{Pv}
P_v=\sum_{e\in E_{ext}(\Gamma), t(e)=v} p_e,
\end{equation}
where the $p_e$ are incoming external momenta attached to the external
edges of $\Gamma$ and satisfying the conservation law
\begin{equation}\label{extmom0}
 \sum_{e\in E_{ext}(\Gamma)} p_e =0. 
\end{equation}

The divergence properties of the integral \eqref{FeynInt} can be
estimated in terms of the ``superficial degree of divergence'', 
which is measured by the quantity $n-D\ell/2$.
The integral \eqref{FeynInt} is called logarithmically divergent when
$n-D\ell/2 =0$. The example of the banana graphs we concentrate on
below has $n=\ell+1$, so that we find $n-D\ell/2=(1-D/2)\ell+1 <0$
for $D> 2$ and $\ell\geq 2$. In this case, we write the integral
\eqref{FeynInt} in the form
\begin{equation}\label{FeynInt2}
U(\Gamma,p)=   \frac{\Gamma(n-D(n-1)/2)}{(4\pi)^{(n-1) D/2}}  
\int_{\sigma_n} \frac{P_\Gamma(p,t)^{-n+D(n-1)/2} \,\omega_n}
{\Psi_\Gamma(t)^{n(-1 +D/2)}},
\end{equation}
where $\omega_n$ is the volume form and 
the domain of integration is the topological simplex
\begin{equation}\label{topSimpl}
\sigma_n=\{ (t_1,\ldots,t_n)\in \R_+^n \,|\, \sum_i t_i =1 \}. 
\end{equation}

The 1PI condition on Feynman graphs comes from the fact of considering
the perturbative expansion of the effective action in quantum field
theory, which reduces the combinatorics of graphs to just those that
are connected and 1PI. In terms of the expression of the Feynman
integral, the 1PI condition is reflected in the fact that only the
propagators for internal edges appear. The parametric form we
described above therefore depends on this assumption. However, for the
algebro-geometric arguments that constitute the main content of this paper, 
the 1PI condition is not strictly necessary.

\subsection{Feynman graphs, varieties, and periods}

The graph polynomial $\Psi_\Gamma(t)$ of \eqref{PsiGamma} also admits
a description as determinant 
\begin{equation}\label{PsiDet}
\Psi_\Gamma(t)=\det M_\Gamma(t)
\end{equation}
of an $\ell\times \ell$-matrix
$M_\Gamma(t)$ associated to the graph (\cite{Naka}, \S 3 and
\cite{BjDr}, \S 18), of the form 
\begin{equation}\label{MGamma}
(M_\Gamma)_{kr}(t)=\sum_{i=1}^n t_i \eta_{ik} \eta_{ir},
\end{equation} 
where the $n \times \ell$-matrix $\eta_{ik}$ is defined in terms of
the edges $e_i \in E(\Gamma)$ and   
a choice of a basis for the first homology group, $l_k \in
H_1(\Gamma,\Z)$,  with $k=1,\ldots, \ell=b_1(\Gamma)$, by setting 
\begin{equation}\label{etaik}
\eta_{ik}=\left\{ \begin{array}{rl} +1 & \text{edge $e_i\in$ loop
$l_k$, same orientation} \\[2mm] -1 & \text{edge $e_i\in$ loop
$l_k$, reverse orientation} \\[2mm] 0 & \text{otherwise,} \end{array}\right.
\end{equation}
after choosing an orientation of the edges.

Notice how the result is independent of the choice of the orientation
of the edges and of the choice of the basis of $H_1(\Gamma,\Z)$. In
fact, a change of orientation in a given edge results in a change of
sign to one of the columns of the matrix $\eta_{ki}$, which is
compensated by the change of sign in the corresponding row of the
matrix $\eta_{ir}$, so that the determinant $\det M_\Gamma(t)$ is
unaffected. Similarly, a change in the choice of the basis of
$H_1(\Gamma,\Z)$ has the effect of changing $M_\Gamma(t) \mapsto A
M_\Gamma(t) A^{-1}$ for some $A\in \GL(\ell,\Z)$ and the determinant
is again unchanged.

\smallskip

The graph hypersurface $X_\Gamma$ is by definition the zero locus of
the Kirchhoff polynomial, 
\begin{equation}\label{XGammaProj}
X_\Gamma =\{ t=(t_1:\ldots :t_n)\in \P^{n-1} \,| \, \Psi_\Gamma(t)=0 \}.
\end{equation}
Since $\Psi_\Gamma$ is homogeneous, it defines a hypersurface in
projective space.

The domain of integration $\sigma_n$ defines a cycle in the relative
homology $H_{n-1}(\P^{n-1},\Sigma_n)$, where $\Sigma_n$ is the algebraic
simplex (the union of the coordinate hyperplanes, see \eqref{algSimpl}
below). The Feynman integral \eqref{FeynInt}, \eqref{FeynInt2} then can be
viewed (\cite{BEK},\cite{Blo}) as the evaluation of an algebraic 
cohomology class in $H^{n-1}(\P^{n-1}\smallsetminus X_\Gamma,
\Sigma \smallsetminus \Sigma\cap X_\Gamma)$ on the cycle
defined by $\sigma_n$. In this sense, it can be viewed as the
evaluation of a period of the algebraic variety given by the
complement of the graph hypersurface. To understand the nature of this
period, one is faced with two main problems. One is 
eliminating divergences (regularization and renormalization of Feynman
integrals), and the other is understanding what kind of motives are
involved in the part of the hypersurface complement
$\P^{n-1}\smallsetminus X_\Gamma$ that is involved in the evaluation
of the period, hence what kind of transcendental numbers one expects
to find in the evaluation of the corresponding Feynman integrals. A
detailed analysis of these problems was carried out in \cite{BEK}. The
examples we concentrate on in this paper are not especially
interesting from the motivic point of view, since they are expressible
in terms of pure Tate motives (\cf \cite{Blo}), but they provide us
with an infinite family of graphs for which all computations are
completely explicit. 

\subsection{Dual graphs and Cremona transformation}\label{CremonaSect}

In the case of {\em planar} graphs, there is an interesting relation
between the hypersurface of the graph and the one of the dual graph.
This will be especially useful in the explicit calculation we perform
below in the special case of the banana graphs. We recall it here in
the general case of arbitrary planar graphs.

The standard Cremona transformation of $\P^{n-1}$ is the map
\begin{equation}\label{Cremona}
\cC: (t_1:\cdots:t_n) \mapsto \left( \frac{1}{t_1}:\cdots:
\frac{1}{t_n}\right). 
\end{equation}
This is a priori defined away from the algebraic simplex of coordinate axes
\begin{equation}\label{algSimpl}
\Sigma_n = \{ (t_1:\cdots:t_n)\in \P^{n-1}\,|\, \prod_i t_i =0 \}\subset
\P^{n-1},
\end{equation}
though we see in Lemma \ref{SandL} below that it is well defined also
on the general point of $\Sigma_n$, its locus of indeterminacies being
only the singularity subscheme of $\Sigma_n$.

Let $\cG(\cC)$ denote the closure of the graph of $\cC$. 
Then $\cG(\cC)$ is a subvariety of $\P^{n-1}\times \P^{n-1}$ with projections
\begin{equation}\label{GraphCrem}
 \xymatrix@C=12pt{
& \cG(\cC) \ar[dl]_{\pi_1} \ar[dr]^{\pi_2} \\
\P^{n-1} \ar@{-->}[rr]^\cC & & \P^{n-1}
} \end{equation}

\begin{lem}\label{CgraphEq}
Using coordinates $(s_1:\cdots:s_n)$ for the target $\P^{n-1}$, the graph
$\cG(\cC)$ has equations
\begin{equation}\label{eqgraphCrem}
t_1 s_1=t_2 s_2=\cdots= t_n s_n  .
\end{equation}
In particular, this describes $\cG(\cC)$ as a complete intersection of
$n-1$ hypersurfaces in $\P^{n-1}\times \P^{n-1}$ with equations 
$t_i s_i=t_n s_n$, for $i=1,\dots,n-1$.
\end{lem}

\proof The equations \eqref{eqgraphCrem} clearly cut out
$\cG(\cC)$ over the open set $\cU\subset \P^n$ where all $t$-coordinates 
are nonzero. Since every component of a scheme defined by $n-1$ 
equations has codimension $\le n-1$, it suffices to show that equations 
\eqref{eqgraphCrem} define a set of codimension~$>n-1$ over the
complement of $\cU$. Now assume that at least one of the
$t$-coordinates equal~$0$. Without
loss of generality, suppose $t_n=0$. Intersecting with the locus
defined by \eqref{eqgraphCrem} determines the set with equations
$$t_1s_1=\dots=t_{n-1} s_{n-1}= t_n=0\quad,$$
which has codimension $n>n-1$, as promised.
\endproof

It is not hard to see that the variety~$\cG(\cC)$ has singularities in 
codimension~$3$. It is nonsingular for $n=2,3$, but singular for $n\ge
4$.

\smallskip

The open set $\cU$ as above is the complement of the divisor $\Sigma_n$ of \eqref{algSimpl}.
The inverse image of $\Sigma_n$ in $\cG(\cC)$ can be described easily.
It consists of the points $$((t_1:\cdots:t_n),(s_1:\cdots:s_n))$$ such that
\[
\{ i \,| \, t_i=0\} \cup \{ j\,| \, s_j=0\} = \{1,\dots,n\}\quad.
\]
This locus consists of $2^N-2$ components of dimension~$n-2$: one 
component for each nonempty proper subset $I$ of $\{1,\dots,n\}$.
The component corresponding to $I$ is the set of points with $t_i=0$
for $i\in I$ and $s_j=0$ for $j\not\in I$.

The situation for $n=3$ is well represented by the famous 
picture of Figure \ref{Cremona3Fig}. 
The three zero-dimensional strata of $\Sigma_3$ are blown up
in $\cG(\cC)$ as we climb the diagram from the lower left to the top. 
The proper transforms of the one dimensional strata are blown down
as we descend to the lower right.
The horizontal rational map is an isomorphism between the complements
of the triangles.
The inverse image of $\Sigma_3$ consists of $2^3-2=6$ components,
as expected.

Of course the situation is completely symmetric: the algebraic simplex
\eqref{algSimpl} may be embedded in the target $\P^n$ as well (with equation
$\prod_i s_i=0$). One has $\pi_1^{-1}(\Sigma_n)=\pi_2^{-1}(\Sigma_n)$.

\smallskip

Let $\cS_n\subset \P^{n-1}$ be the subscheme defined by the ideal
\begin{equation}\label{idealS}
\cI_{\cS_n}=(t_1\cdots t_{n-1}, t_1\cdots t_{n-2} t_n, \dots,
t_1 t_3 \cdots t_n, t_2 \cdots t_n).
\end{equation}
The scheme $\cS_n$ is the singularity subscheme of the divisor with simple normal
crossings $\Sigma_n$ of \eqref{algSimpl}, given by the union of the
coordinate hyperplanes. We can place $\cS_n$ in both the source and
target $\P^{n-1}$. 
Finally, let $\cL$ be the hyperplane defined by the equation
\begin{equation}\label{Lhyp}
\cL=\{ (t_1:\cdots:t_n)\in \P^{n-1}\,|\, t_1+\cdots+t_n =0 \}.
\end{equation}

\begin{center}
\begin{figure}
\includegraphics[scale=.4]{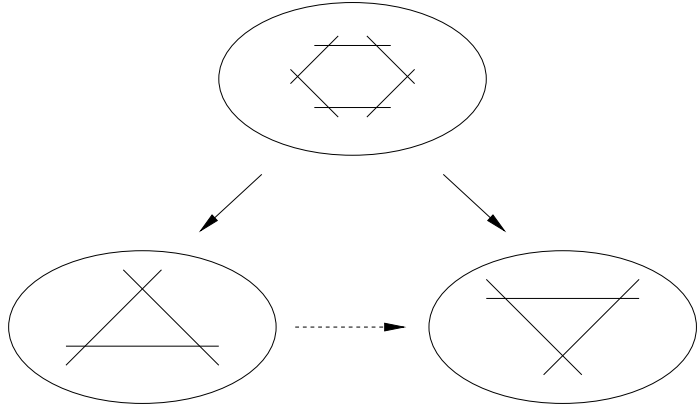}
\caption{The Cremona transformation in the case
$n=3$. \label{Cremona3Fig}}
\end{figure}
\end{center}

We then can make the following observations.

\begin{lem}\label{SandL}
Let $\cC$, $\cG(\cC)$, $\cS_n$, and $\cL$ be as above. 
\begin{enumerate}
\item $\cS_n$ is the subscheme of indeterminacies of the Cremona
transformation $\cC$.
\item $\pi_1:\cG(\cC)\to \P^{n-1}$ is the blow-up along $\cS_n$.
\item $\cL$ intersects every component of $\cS_n$ transversely.
\item $\Sigma_n$ cuts out a divisor with simple normal crossings on $\cL$.
\end{enumerate}
\end{lem}

\proof (1) Notice that the definition \eqref{Cremona} of the Cremona
transformation, which is a priori defined on the complement of
$\Sigma_n$ still makes sense on the general point of $\Sigma_n$.
Thus, the indeterminacies of the map \eqref{Cremona} are contained in
the singularity locus $\cS_n$ of $\Sigma_n$ defined by
\eqref{idealS}. It consists in fact of all of $\cS_n$ since 
after `clearing denominators', the components of the map defining $\cC$ 
given in \eqref{Cremona} can be rewritten as:
\begin{equation}\label{cleardenom}
(t_1:\cdots: t_n) \mapsto (t_2 \cdots t_n : t_1 t_3 \cdots t_n : \cdots :
t_1\cdots t_{n-1})\quad ,
\end{equation}
so that one sees that the indeterminacies are precisely those defined
by the ideal \eqref{idealS}.

(2) Using \eqref{cleardenom}, the map $\pi_1: \cG(\cC)\to \P^n$ may be 
identified with the blow-up of $\P^n$ along the subscheme $\cS_n$ 
defined by the ideal $\cI_{\cS_n}$ of \eqref{idealS}.
The generators of this ideal are the partial derivatives of the equation
of the algebraic simplex. Thus, $\cS_n$ is the {\em singularity subscheme\/}
of $\Sigma_n$. It consists of the union of the closure of the 
dimension~$n-2$ strata of $\Sigma_n$.
Again, note that the situation is entirely symmetrical: we can place $\cS_n$
in the target $\P^n$ as well, and view $\pi_2$ as the blow-up along
$\cS_n$.

(3) and (4) are immediate from the definitions.
\endproof

\begin{center}
\begin{figure}
\includegraphics[scale=0.5]{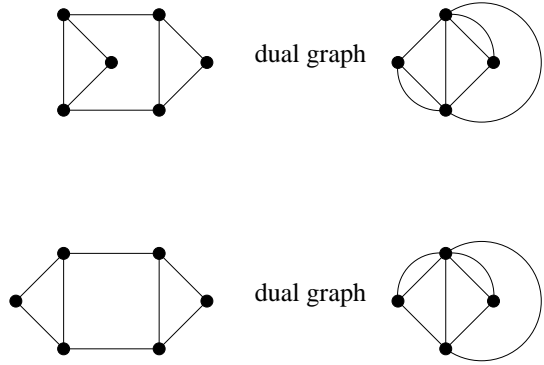}
\caption{Dual graphs of different planar embeddings of the same graph.
\label{DualGraphsFig}}
\end{figure}
\end{center}

Given a connected planar graph $\Gamma$, one defines its {\em dual graph}
$\Gamma^\vee$ by fixing an embedding of $\Gamma$ in $\R^2\cup \{
\infty\}=S^2$ and constructing a new graph in $S^2$ that has a vertex in each
component of $S^2\smallsetminus \Gamma$ and one edge connecting two
such vertices for each edge of $\Gamma$ that is in the common boundary of
the two regions containing the vertices. Thus, $\# E(\Gamma^\vee)=\#
E(\Gamma)$ and $\# V(\Gamma^\vee)=b_0(S^2\smallsetminus \Gamma)$. The
dual graph is in general non-unique, since it depends on the choice of
the embedding of $\Gamma$ in $S^2$, see \eg Figure \ref{DualGraphsFig}.

We recall here a well known result (see \eg \cite{Blo}, Proposition
8.3), which will be very useful in the following.

\begin{lem}\label{PsiDual}
Suppose given a planar graph $\Gamma$ with $\#E(\Gamma)=n$, with dual
graph $\Gamma^\vee$. Then the graph polynomials satisfy  
\begin{equation}\label{graphPolDual}
\Psi_\Gamma(t_1,\ldots,t_n)= 
  (\prod_{e\in E(\Gamma)} t_e)\,\, \Psi_{\Gamma^\vee}(t_1^{-1},\ldots,
t_n^{-1}),  
\end{equation}
hence the graph hypersurfaces are related by the Cremona
transformation $\cC$ of \eqref{Cremona}, 
\begin{equation}\label{XGammaDual}
 \cC(X_\Gamma \cap (\P^{n-1}\smallsetminus \Sigma_n)) =
X_{\Gamma^\vee}\cap (\P^{n-1}\smallsetminus \Sigma_n).
\end{equation}
\end{lem}

\proof This follows from the combinatorial identity
$$ \begin{array}{rl}
\Psi_\Gamma(t_1,\ldots,t_n)= & \sum_{T\subset\Gamma} \prod_{e\notin E(T)}
t_e \\[2mm] = & (\prod_{e\in E(\Gamma)} t_e) \sum_{T\subset\Gamma}
\prod_{e\in E(T)} t_e^{-1} \\[2mm]
= & (\prod_{e\in E(\Gamma)} t_e) \sum_{T'\subset \Gamma^\vee}
\prod_{e\notin E(T')} t_e^{-1} \\[2mm]  
= &  (\prod_{e\in E(\Gamma)} t_e) \Psi_{\Gamma^\vee}(t_1^{-1},\ldots,
t_n^{-1}) .
\end{array} $$
The third equality uses the fact that $\# E(\Gamma)=\#E(\Gamma^\vee)$
and $\# V(\Gamma^\vee)=b_0(S^2\smallsetminus \Gamma)$, so that
$\deg \Psi_\Gamma + \deg \Psi_{\Gamma^\vee} = \# E(\Gamma)$, and the
fact that there is a bijection between complements of spanning tree 
$T$ in $\Gamma$ and spanning trees $T'$ in $\Gamma^\vee$ obtained by
shrinking the edges of $T$ in $\Gamma$ and taking the dual graph of
the resulting connected graph. 

Written in the coordinates $(s_1:\cdots:s_n)$ of the target $\P^{n-1}$
of the Cremona transformation, the identity \eqref{graphPolDual} gives
$$ \Psi_\Gamma(t_1,\ldots,t_n)= (\prod_{e\in E(\Gamma^\vee)} s_e^{-1})
\Psi_{\Gamma^\vee}(s_1,\ldots, s_n) $$
from which \eqref{XGammaDual} follows.
\endproof

We then have the following simple geometric observation, which follows
directly from Lemma \ref{SandL} and Lemma \ref{PsiDual} above. 

\begin{cor}\label{DualGrCoroll}
The graph hypersurface of the dual graph is $X_{\Gamma^\vee}=\pi_2
(\pi_1^{-1}(X_{\Gamma}))$, with $\pi_i: \cG(\cC)\to \P^{n-1}$, for $i=1,2$, 
as in \eqref{GraphCrem}. The Cremona transformation $\cC$ restricts
to a (biregular) isomorphism 
\begin{equation}\label{isoCremDual}
\cC: X_\Gamma \smallsetminus \Sigma_n
\to X_{\Gamma^\vee} \smallsetminus \Sigma_n.
\end{equation}
The map $\pi_2:\cG(\cC)\to \P^{n-1}$ of \eqref{GraphCrem} restricts to
an isomorphism 
\begin{equation}\label{isoCrem2Dual}
\pi_2: \pi_1^{-1}(X_\Gamma \smallsetminus \Sigma_n)  
\to X_{\Gamma^\vee} \smallsetminus \Sigma_n.
\end{equation}
\end{cor}

\smallskip

Notice that the formula \eqref{graphPolDual} can be used as a 
source of examples of
combinatorially inequivalent graphs that have the same graph
hypersurface. In fact, the graph polynomial
$\Psi_{\Gamma^\vee}(s_1,\ldots, s_n)$ is the same independently of the
choice of the embedding of the planar graph $\Gamma$ in the
plane, while the dual graph $\Gamma^\vee$ depends on the choice of the
embedding of $\Gamma$ in the plane. Thus, different embeddings that
give rise to different graphs $\Gamma^\vee$ provide examples of
combinatorially inequivalent graphs with the same graph hypersurface.
This has direct consequences, for example, on the question of lifting
the Connes--Kreimer Hopf algebra of graphs \cite{CK} at the level of the graph
hypersurfaces or their classes in the Grothendieck ring of motives.
An explicit example of combinatorially inequivalent graphs with the
same graph hypersurface, obtained as dual graphs of different planar
embeddings of the same graph, is given in Figure \ref{DualGraphsFig}.

\smallskip

We see a direct application of this general result for planar graphs
in \S \ref{DualClassSec} below, where we derive a relation between the
classes in the Grothendieck ring. In general, this relation alone is
too weak to give explicit formulae, but the example we concentrate on
in the next section shows a family of graphs for which a complete
description of both the class in the Grothendieck ring and the CSM
class follows from the special form that the result of Corollary
\ref{DualGrCoroll} takes.

\begin{center}
\begin{figure}
\includegraphics[scale=0.5]{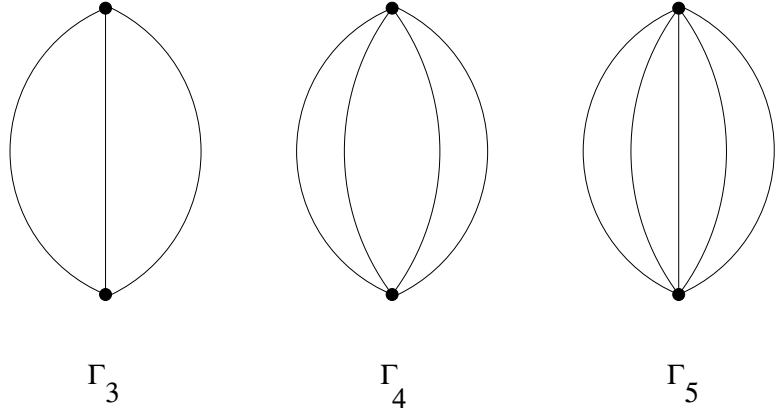}
\caption{Examples of banana graphs
\label{BananaFig}}
\end{figure}
\end{center}

\subsection{An example: the banana graphs}\label{BanSec}

In this paper we concentrate on a particular example, for which
we can carry out complete and explicit calculations. We consider 
an infinite family of graphs called the ``banana graphs''. 
The $n$-th term $\Gamma_n$ in this family is a vacuum bubble Feynman 
graph for a scalar field theory with an interaction term of the 
form $\cL_{int}(\phi)=\phi^n$. The graph $\Gamma_n$ has two
vertices and $n$ parallel edges between them, as in Figure
\ref{BananaFig}. 

A direct computation using the Macaulay2 program \cite{Macaulay2} for characteristic
classes developed in \cite{Alu3} shows, for the first three examples
in this series of graphs depicted in Figure \ref{BananaFig}, the
following invariants (see \S \ref{GrCSMdefSec} for precise definitions).

\begin{center}
\begin{tabular}{|c|c|c|c|}
\hline
$n$ & $3$ & $4$ & $5$ \\
\hline
& & & \\ 
$\Psi_\Gamma$ & $t_1t_2+t_2t_3+t_1t_3$
 &  $t_1t_2t_3+t_1t_2t_4 +$ &  $t_1t_2t_3t_4+ t_1t_2t_3t_5 + $  \\ 
& & $t_1t_3t_4+t_2t_3t_4$  & $t_1t_2t_4t_5+t_1t_3t_4t_5+t_2t_3t_4t_5$ \\
& & &  \\
\hline
& & & \\
$c(X_\Gamma)$ & $2H^2+2H$ & $5H^3 + 3H^2 + 3H$  &
$4H^4  + 14H^3+ 4H^2  + 4H$   \\
& & & \\
\hline
& & & \\
$Mil(X_\Gamma)$ & $0$ & $-4 H^3$ & $60 H^4 -10 H^3$ \\
& & & \\
\hline
& & & \\
$\chi(X_\Gamma)$ & $2$ & $5$ & $4$ \\
& & & \\
\hline
\end{tabular}
\end{center}

\bigskip

Here $H$ denotes the hyperplane class and $c(X_\Gamma)$ is the
Chern--Schwartz--MacPherson class of the hypersurface pushed forward
to the ambient projective space. We also show the Milnor class, which
measures the discrepancy between the Chern--Schwartz--MacPherson class 
and the Fulton class, that is, between the characteristic class of the
singular hypersurface $X_\Gamma = \{ \Psi_\Gamma =0 \}$ and the class
of a smooth deformation. We also display the value of the Euler
characteristic, which one can read off the CSM class.
The reader can pause momentarily to consider the CSM classes reported
in the three examples above and notice that they suggest a 
general formula for this family of graphs, where the coefficient of
$H^k$ in the CSM class for the $n$-th hypersurface $X_{\Gamma_n}$ is
given by the formula
\begin{equation}\label{CSMguess}
\left\{\aligned
& \binom {n}{k}-\binom{n-1}k=\binom {n-1}{k-1}\quad\text{if $k$ is even} \\
& \binom {n}{k}+\binom{n-1}k\phantom{=\binom {n-1}{k-1}\,\,}\quad\text{if $k$ is odd}
\endaligned\right.
\end{equation}
for $1<k<n$, and $n-1$ for $k=1$. Thus, for example, for
$n\ge 3$ the Euler characteristic $\chi(X_{\Gamma_n})$ of the $n$-th
banana hypersurface fits the pattern
\begin{equation}\label{chiBn}
\chi(X_{\Gamma_n})= n+(-1)^n.
\end{equation}
This is indeed the correct formula for the CSM class that will be
proved in \S \ref{CSMsect} below. The sample case reported here
already exhibits an interesting feature, which we encounter again in
the general formula of \S  \ref{CSMsect} and which seems confirmed by 
computations carried out algorithmically on other sample
graphs from different families of Feynman graphs, namely the
unexpected {\em positivity} of the coefficients of the
Chern--Schwartz--MacPherson classes. Notice that a similar instance of
positivity of the CSM classes arises in another case of varieties with
a strong combinatorial flavor, namely the case of the Schubert
varieties considered in \cite{AluMih}. At present we do not have a
conceptual explanation for this positivity phenomenon, but we can
state the following tentative guess, based on the sparse 
numerical and theoretical evidence gathered so far.

\begin{conj}\label{posconj}
The coefficients of all the powers $H^k$ in the CSM class of an
arbitrary graph hypersurface $X_\Gamma$ are non-negative.
\end{conj}

\smallskip

For the general element $\Gamma_n$ in the family of the banana graphs,
the graph hypersurface $X_{\Gamma_n}$ in $\P^{n-1}$ is defined by the 
vanishing of the graph polynomial 
\begin{equation}\label{PolyBn}
\Psi_{\Gamma_n} = t_1\cdots t_n (\frac{1}{t_1} + \cdots + \frac{1}{t_n}).
\end{equation}
This is easily seen, since in this case spanning trees consist of a
single edge connecting the two vertices. 
Equivalently, one can see this in terms of the matrix
$M_\Gamma(t)$. 

\begin{lem}\label{MGamman}
For the $n$-th banana graph $\Gamma_n$, the matrix $M_{\Gamma_n}(t)$
is of the form
\begin{equation}\label{nMGamma}
M_{\Gamma_n}(t)=\left( \begin{array}{cccccc}
t_1+t_2 & -t_2 & 0 & 0 & \cdots & 0 \\
-t_2 & t_2+t_3 & -t_3 & 0 &  & 0 \\
0 & -t_3 & t_3+t_4 & -t_4 &  & 0 \\
0 & 0 & -t_4 & t_4 + t_5 & & 0 \\
\vdots & \vdots & & & & \vdots \\
0 & 0 & 0 & 0  & \cdots & t_{n-1}+t_n 
\end{array}\right).
\end{equation}
\end{lem}

\proof In fact, if we choose as a basis of the first
cohomology of the graph $\Gamma_n$ the obvious one consisting 
of the $\ell=n-1$ loops $e_i \cup -e_{i+1}$, with $i=1,\ldots,n-1$, we
obtain that the $n\times (n-1)$-matrix $\eta_{ik}$ is of the form
$$ \eta_{ik} =\left( \begin{array}{rrrrrr}
1&0&0&0&0&\cdots\\-1&1&0&0&0&\cdots\\0&-1&1&0&0&\cdots\\
0&0&-1&1&0&\cdots \\ 0&0&0&-1&1&\cdots 
\end{array}\right).$$
Thus, the matrix $(M_\Gamma)_{rk}(t)=\sum_i t_i \eta_{ri}\eta_{ik}$
has the form \eqref{nMGamma}. It is easy to check that this indeed has
determinant given by \eqref{PolyBn}. In fact, from \eqref{nMGamma} one
sees that the determinant satisfies
$$ \det M_{\Gamma_n} (t)= (t_{n-1}+t_n)\, \det M_{\Gamma_{n-1}}(t) \,  -
t_{n-1}^2\, \det M_{\Gamma_{n-2}}(t) . $$
It then follows by induction that the determinant satisfies the
recursive relation
\begin{equation}\label{detPsiInd}
\det M_{\Gamma_n}(t)= t_n \, \det
M_{\Gamma_{n-1}}(t) \, + t_1 \cdots t_{n-1}. 
\end{equation}
In fact, assuming the above for $n-1$ we obtain
$$ \det M_{\Gamma_n}(t)= t_n \det M_{\Gamma_{n-1}}(t) + t_{n-1}^2
\det M_{\Gamma_{n-2}}(t) + t_1 \cdots t_{n-1} - t_{n-1}^2
\det M_{\Gamma_{n-2}}(t). $$
It is then clear that $\det M_{\Gamma_n}(t)=\Psi_{\Gamma_n}(t)$, with 
the latter given by the formula \eqref{PolyBn}, since this also
clearly satisfies the same recursion \eqref{detPsiInd}.
\endproof

The dual graph $\Gamma^\vee_n$ is just a polygon with
$n$ vertices and we can identify the hypersurface $X_{\Gamma^\vee_n}$
in $\P^{n-1}$ with the hyperplane $\cL$ defined in \eqref{Lhyp}.

We rephrase here the statement of Corollary \ref{DualGrCoroll} in the
special case of the banana graphs, since it will be very useful in 
our explicit computations of \S\S \ref{MotSect} and \ref{CSMsect} below. 

\begin{lem}\label{BananaDual}
The $n$-th banana graph hypersurface is $X_{\Gamma_n}=\pi_2
(\pi_1^{-1}(\cL))$, with $\pi_i: \cG(\cC)\to \P^{n-1}$, for $i=1,2$, 
as in \eqref{GraphCrem}. The Cremona transformation $\cC$ restricts
to a (biregular) isomorphism 
\begin{equation}\label{isoCrem}
\cC: \cL\smallsetminus \Sigma_n
\to X_{\Gamma_n} \smallsetminus \Sigma_n.
\end{equation}
The map $\pi_2:\cG(\cC)\to \P^{n-1}$ of \eqref{GraphCrem} restricts to
an isomorphism 
\begin{equation}\label{isoCrem2}
\pi_2: \pi_1^{-1}(\cL\smallsetminus \Sigma_n)  \to X_{\Gamma_n} \smallsetminus \Sigma_n.
\end{equation}
\end{lem}

\begin{center}
\begin{figure}
\includegraphics[scale=0.4]{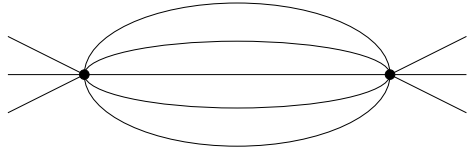}
\caption{Banana graphs with external edges
\label{BnExtEdgeFig}}
\end{figure}
\end{center}

In order to compute the Feynman integral \eqref{FeynInt2}, we view the
banana graphs $\Gamma_n$ not as vacuum bubbles, but as endowed with a
number of external edges, as in Figure \ref{BnExtEdgeFig}. It does not
matter how many external edges we attach. This will depend on which
scalar field theory the graph belongs to, but the resulting integral
is unaffected by this, as long as we have nonzero external momenta
flowing through the graph. 

\begin{lem}\label{BnFeyn}
The Feynman integral \eqref{FeynInt2} for
the banana graphs $\Gamma_n$ is of the form
\begin{equation}\label{BnFeynInt}
U(\Gamma,p)=   \frac{\Gamma((1-D/2)(n-1) +1) C(p)}{(4\pi)^{(n-1) D/2}}  
\int_{\sigma_n} \frac{(t_1\cdots t_n)^{(\frac{D}{2}-1)(n-1) -1} \,\omega_n}
{\Psi_\Gamma(t)^{(\frac{D}{2}-1)n}},
\end{equation}
with the function of the external momenta given by
$C(p)= (\sum P_v)^2$, with $v$ being either one of the two vertices of
the graph $\Gamma_n$ and  $P_v=\sum_{e\in E_{ext}(\Gamma_n), t(e)=v} p_e$. 
\end{lem}

\proof The result is immediate from \eqref{FeynInt2}, using $n=\ell+1$
and the fact that the only cut-set for the banana graph $\Gamma_n$
consists of the union of all the edges, so that 
$$ P_\Gamma(t,p)=C(p) \, t_1\cdots t_n. $$
\endproof

For example, in the case with $n=2$ and $D\in 2\N$, $D\geq 4$, the integral (up
to a divergent Gamma factor $\Gamma(2-D/2) 4\pi^{-D/2}$) reduces
to the computation of the convergent integral
$$ \int_{[0,1]} (t (1-t))^{D/2-2} dt =
\frac{((\frac{D}{2}-2)!)^2}{(D-3)!}. $$ 

\smallskip

In general, apart from poles of the Gamma function, divergences may arise from
the intersections of the domain of integration $\sigma_n$ with the
graph hypersurface $X_{\Gamma_n}$. 

\begin{lem}\label{intDomXGamma}
The intersection of the domain of integration $\sigma_n$ with the
graph hypersurface $X_{\Gamma_n}$ happens along $\sigma_n \cap \cS_n$ 
in the algebraic simplex $\Sigma_n$. 
\end{lem}

\proof The polynomial $\Psi_\Gamma(t)\geq 0$ for $t\in \R_+^n$ and by
the explicit form \eqref{PolyBn} of the polynomial, one can see that 
zeros will only occur when at least two of the coordinates vanish, \ie
along the intersection of $\sigma_n$ with the scheme of singularities 
$\cS_n$ of $\Sigma_n$ (\cf Lemma \ref{XGammaS} below).
\endproof

One procedure to deal with this source of divergences is to work on
blowups of $\P^{n-1}$ along this singular locus (\cf \cite{BEK}, \cite{Blo}). 
In \cite{Mar} another possible method of regularization for integrals
of the form \eqref{BnFeynInt} which takes care of the singularities of the
integral on $\sigma_n$ (the pole of the Gamma function needs to be
addressed separately) was proposed, based on replacing the integral
along $\sigma_n$ with an integral that goes around the singularities
along the fibers of a circle bundle. In general, this type of
regularization procedures requires a detailed knowledge of the 
singularities of the hypersurface $X_\Gamma$ to be carried out, 
and that is one of the reasons for introducing invariants of singular 
varieties in the study of graph hypersurfaces.

\section{Characteristic classes and the Grothendieck ring}\label{GrCSMdefSec}

In order to understand the nature of the part of the cohomology of the
graph hypersurface complement that supports the period corresponding
to the Feynman integral (ignoring divergence issues momentarily), one
would like to decompose $\P^{n-1}\smallsetminus X_{\Gamma}$ into
simpler building blocks. As in \S 8 of \cite{BEK}, this can be done by
looking at the class $[X_\Gamma]$ of the graph hypersurface in the
Grothendieck ring of motives. One knows by the general result of
Belkale--Brosnam \cite{BeBro} that the graph hypersurfaces generate
the Grothendieck ring, hence they are quite arbitrarily complex as
motives, but one still needs to understand whether the part of the
decomposition that is relevant to the computation of the Feynman
integral might in fact be of a very special type, \eg a mixed Tate
motive as the evidence suggests. The family of graphs we consider here
is very simple in that respect. In fact, one can see very explicitly
that their classes in the Grothendieck ring are combinations of  
Tate motives (\cf the formula \eqref{clXGammaGrL} below). One can see
this also by looking at the Hodge structure. For the graph
hypersurfaces of the banana graphs this is described in \S 8 of
\cite{Blo}. 

Here we describe two ways of analyzing the graph hypersuraces through
an additive invariant, one as above using the class $[X_\Gamma]$ in the
Grothendieck ring, and the other using the pushforward of the
Chern--Schwartz--MacPherson class of $X_{\Gamma}$ to the Chow group 
(or homology) of the ambient projective space $\P^{n-1}$. 
While the first does not depend on an
ambient space, the latter is sensitive to the specific embedding of
$X_{\Gamma}$ in the projective space $\P^{n-1}$, hence it might
conceivably carry a little more information that is useful in relation
to the computation of the Feynman integral on $\P^{n-1}\smallsetminus
X_{\Gamma}$. We recall here below a few basic facts about both
constructions. The reader familiar with these generalities can 
skip directly to the next section.

\subsection{The Grothendieck ring}\label{GrDefSec}

Let $\cV_K$ denote the category of algebraic varieties over a field $K$.
The Grothendieck ring $K_0(\cV_K)$ is the abelian group generated by
isomorphism classes $[X]$ of varieties, with the relation
\begin{equation}\label{relGrV}
[X]=[Y]+[X\smallsetminus Y],
\end{equation}
for $Y\subset X$ closed. It is made into a ring by the product
$[X\times Y]=[X][Y]$. 

An {\em additive invariant} is a map $\chi : \cV_K \to R$, with values in a
commutative ring $R$, satisfying $\chi(X)=\chi(Y)$ if $X\cong Y$ are
isomorphic, $\chi(X)=\chi(Y)+\chi(X\smallsetminus Y)$ for $Y\subset X$
closed, and $\chi(X\times Y)=\chi(X)\chi(Y)$. The Euler characteristic
is the prototype example of such an invariant. Assigning an
additive invariant with values in $R$ is equivalent to assigning 
a ring homomorphism $\chi: K_0(\cV_K) \to R$. 

Let $\cM_K$ be the pseudo-abelian category of (Chow) motives over $K$.
We write the objects of $\cM_K$ in the form $(X,p,m)$, with $X$ a
smooth projective variety over $K$, $p=p^2\in \End(X)$ a projector,
and $m\in\Z$ accounting for the twist by powers of the Tate motive
$\Q(1)$. Let $K_0(\cM_K)$ denote the Grothendieck ring of the category
$\cM_K$ of motives. 
The results of \cite{GilSou} show that, for $K$ of characteristic zero,
there exists an additive invariant $\chi: \cV_K \to K_0(\cM_K)$. This
assigns to a smooth projective variety $X$ the class
$\chi(X)=[(X,id,0)] \in K_0(\cM_K)$, while for $X$ a general variety 
it assigns a complex $W(X)$ in the category of complexes over $\cM_K$,
which is homotopy equivalent to a bounded complex whose class in
$K_0(\cM_K)$ defines the value $\chi(X)$. This defines a ring
homomorphism 
\begin{equation}\label{GrVGrMot}
\chi: K_0(\cV_K) \to K_0(\cM_K).
\end{equation}
If $\bL$ denotes the class $\bL=[\A^1] \in K_0(\cV_K)$ then its image
in $K_0(\cM_K)$ is the Lefschetz motive
$\bL=\Q(-1)=[(\Spec(K),id,-1)]$. Since the Lefschetz motive is
invertible in $K_0(\cM_K)$, its inverse being the Tate motive $\Q(1)$,
the ring homomorphism \eqref{GrVGrMot} induces a ring homomorphism
\begin{equation}\label{GrVGrMot2}
\chi: K_0(\cV_K)[\bL^{-1}] \to K_0(\cM_K).
\end{equation}

Thus, in the following we can either regard the classes $[X_\Gamma]$
of the graph hypersurfaces in the Grothendieck ring of varieties
$K_0(\cV_K)$ or, under the homomorphism \eqref{GrVGrMot}, as elements
in the Grothendieck ring of motives $K_0(\cM_K)$. We will no longer make
this distinction explicit in the following.

\subsection{CSM classes as a measure of
singularities} \label{CSMdefSec}

The {\em Chern class\/} of a nonsingular complete variety $V$ is the 
`total homology Chern class' of its tangent bundle. We write
$c(V):=c(TV)\cap [V]_*$ to indicate the result of applying the Chern
class of the tangent bundle of $V$ to the fundamental class $[V]_*$
of $V$. (We use the notation $[V]_*$ rather than the more common $[V]$
in order to avoid any confusion with the class of $V$ in the Grothendieck
group.)

The class $c(V)$ resides naturally in the Chow group $A_*V$. For the 
purpose of this paper, the reader will miss nothing by replacing $A_*V$
with ordinary homology.

The Chern class of a variety $V$ is a class of evident geometric 
significance: for example, the degree
of its zero-dimensional component agrees with the topological Euler 
characteristic of $V$. This follows essentially from 
the Poincar\'e-Hopf theorem:
\[
\int c(TV)\cap [V]_* = \chi(V)\quad.
\]

It is natural to ask whether there are analogs of the Chern class defined
for possibly {\em singular\/} varieties, for which a tangent bundle is not 
necessarily available.

Somewhat surprisingly, one finds that there are several possible definitions, 
each `natural' for different reasons, and all agreeing with each other in
the nonsingular case. If $X$ is a complete intersection in a nonsingular 
variety $V$, it is reasonable to consider the Fulton class
\[
c_{vir}(X):=\frac{c(TV)}{c(N_XV)}\cap [X]_*\quad,
\]
where $N_XV$ denotes the normal bundle to $X$ in $V$. Up to natural
identifications, this is the Chern class of a smoothing of $X$ (when a
smoothing exists), and in particular it agrees with $c(X)$ if $X$ is 
nonsingular. It is an interesting fact that this class is independent of the 
realization of $X$ as a complete intersection: that is, it is independent
of the ambient nonsingular variety~$V$. In other words, $\frac{c(TV)}
{c(N_XV)}$ behaves as the class of a `virtual tangent bundle' to $X$.
Its definition can in fact be extended (and in more than one way) to 
arbitrary varieties, see~\S 4.2.6 in~\cite{Fulton}.

The class $c_{vir}(X)$ is in a sense unaffected by the 
singularities of $X$: for a hypersurface $X$ in a nonsingular variety~$V$,
it is determined by the class of $X$ as a divisor in $V$. 

A much more refined invariant is the {\em Chern-Schwartz-MacPherson\/}
(CSM) class of $X$, which depends more crucially on the singularities
of $X$, and which we will use as a measure of the singularities by
comparison with $c_{vir}(X)$.

The name of the class retains some of its history. In the mid-60s,
M.-H.~Schwartz (\cite{MHSchwartz1}, \cite{MHSchwartz2}) introduced a 
class extending to singular varieties Poincar\'e-Hopf-type results, 
by studying tangent frames emanating radially from the singularities. 
Independently of Schwartz' work, Grothendieck and Deligne 
conjectured a theory of characteristic classes fitting a tight functorial 
prescription, and in the early 70s R.~MacPherson constructed a class 
satisfying this requirement (\cite{MacPher}). It was later proved
by J.-P.~Brasselet and M.-H.~Schwartz (\cite{BraSchwa}) that the 
classes agree.

In this paper we denote the Chern-Schwartz-MacPherson class of a 
singular variety $X$ simply by $c(X)$ (the notation $c_{SM}(X)$ is
frequently used in the literature).

The properties satisfied by CSM classes may be summarized as
follows. First of all, $c(X)$ must agree with its namesake when $X$
is a complete nonsingular variety: that is, $c(X)=c(TX)\cap [X]_*$
in this case. Secondly, associate with every variety $X$ an abelian
group $F(X)$ of `constructible functions': elements of $F(X)$ are
finite integer linear combinations of functions $\one_Z$ (defined
by $\one_Z(p)=1$ if $p\in Z$, $\one_Z(p)=0$ if $p\not\in Z$), for
subvarieties $Z$ of $X$. The assignment $X\mapsto F(X)$ is
covariantly functorial: for every proper map $Y\to Z$ there is a
push-forward $f_*: F(Y) \to F(X)$, defined by taking topological
Euler characteristic of fibers. More precisely, for $W\subseteq Y$
a closed subvariety, one defines $f_*(\one_W)=
\chi(W\cap f^{-1}(p))$, and extends this definition to $F(Y)$ by
linearity.

Grothendieck and Deligne conjectured the existence of a unique
natural transformation $c_*$ from the functor $F$ to the homology
functor such that $c_*(\one_X)=c(TX)\cap [X]_*$ if $X$ is nonsingular.
MacPherson constructed such a transformation in \cite{MacPher}.
The CSM class of $X$ is then defined to be $c(X):=c_*(\one_X)$.
Resolution of singularities in characteristic zero 
implies that the transformation is unique,
and in fact determines $c(X)$ for any $X$.

As an illustration of the fact that the CSM class assembles interesting
invariants of a variety, apply the property just reviewed to the constant
map $f: X \to \{pt\}$. 
In this case, the naturality property reads $f_* c_*(\one_X) = 
c_* f_*(\one_X)$, that is,
\[
f_* c(X) = c_* (\chi(X) \one_{pt})
\]
(using the definition of push-forward of constructible function). Taking
degrees, this shows that
\[
\int c(X) = \chi(X)\quad,
\]
precisely as in the nonsingular case: the degree of the CSM class
of a (possibly) singular variety equals its topological Euler
characteristic. 

It follows that, if $X$ is a hypersurface with one isolated
singularity, then the degree of the class
\[
Mil(X):=c(X)-c_{vir}(X)
\]
equals (up to a sign) the {\em Milnor number\/} of the singularity. 

For hypersurfaces with arbitrary singularities, as the graph
hypersurfaces we consider in the present paper typically are, the
degree of the CSM class equals Parusi\'nski's generalization of the
Milnor number, \cite{Parus}. The class $Mil(X)$ is
called `Milnor class', and has been studied rather carefully for $X$
a complete intersection, \cite{BLSS}. 

For a hypersurface, the Milnor class carries essentially the same 
information as the Segre class of the singularity subscheme of $X$ 
(see \cite{Alu5}). In this sense, it is a measure of the singularities of
the hypersurface. For example, the largest dimension of a nonzero
term in the Milnor class equals the dimension of the singular
locus of $X$.

The graph hypersurfaces in this paper are hypersurfaces of projective
space, hence it is convenient to view the CSM class and the Milnor class
of $X$ as classes in projective space. This pushforward is understood
in the table in \S \ref{BanSec}, and will be often understood in the 
explicit computations of \S \ref{CSMsect}. 

\subsection{CSM classes versus classes in the
Grothendieck ring}\label{CSMvsGrSec}

CSM classes are defined in \cite{MacPher} by relating them to
a different class, called `Chern-Mather class', by means of a local
invariant of singularities known as the `local Euler obstruction'.
As noted above, once the existence of the classes has been established,
then their computation may be performed by systematic use of 
resolution of singularities and computations of Euler characteristics
of fibers. 

The following direct construction streamlines such computations, by
avoiding any computation of local invariants or of Euler
characteristics. This is observed in \cite{Alu06a} and \cite{Alu2}, 
where it is used to provide an alternative proof of the Grothendieck-Deligne 
conjecture, and as the basis of a generalization of the functoriality
of CSM classes to possibly non-proper morphisms. 

Given a variety $X$, let $Z_i$ be a finite collection of {\em locally closed,
nonsingular\/} subvarieties such that $X = \amalg_i Z_i$.
For each $i$, let $\nu_i: W_i \to \overline Z_i$ be a resolution of 
singularities of the closure of $Z_i$ in $X$, such that the complement
$\overline Z_i\smallsetminus Z_i$ pulls back to a divisor with normal
crossings $E_i$ on $W_i$. Then
\[
c(X)=\sum_i \nu_{i*} c(TW_i(-\log E_i))\cap [W_i]_*\quad.
\]
Here the bundle $TW_i(-\log E_i)$ is the dual of the bundle $\Omega^1_{W_i}
(\log E_i)$ of differential forms on $W_i$ with logarithmic poles along
$E_i$. Each term
\begin{equation}\label{nuiTW}
\nu_{i*} c(TW_i(-\log E_i))\cap [W_i]_*
\end{equation}
computes the contribution $c(\one_{Z_i})$ to the CSM class of $X$
due to the (possibly) non-compact subvariety $Z_i$.

We will use this formulation in terms of duals of sheaves of forms
with logarithmic poles to obtain the results of \S \ref{CSMsect} below. 

By abuse of notation, we denote by $c(Z)\in A_*V$ the class so defined,
for any locally closed subset $Z$ of a large ambient variety $V$.
With this notion in hand, note that if $Y\subseteq X$ are (closed)
subvarieties of $V$, then
\[
c(X)=c(Y)+c(X\smallsetminus Y)\quad,
\]
where push-forwards are, as usual, understood. This relation is
very reminiscent of the basic relation \eqref{relGrV} that 
holds in the Grothendieck group of varieties (see~\S \ref{GrDefSec}). 
At the same time, CSM classes satisfy a `product formula' analogous to
the definition of product in the Grothendieck ring
(\cite{Kwiec}, \cite{Alu06a}).

Moreover, CSM classes satisfy an
`embedded inclusion-exclusion' principle. Namely, if $X_1$ and $X_2$ are
subvarieties of a variety $V$, then
\[
c(X_1\cup X_2) = c(X_1) + c(X_2) - c(X_1\cap X_2) .
\]
This is clear both from the construction presented above and
from the basic functoriality property.

In short, there is an intriguing parallel between operations in 
the Grothendieck group of varieties and manipulations of
CSM classes. This parallel cannot be taken too far, since
the `embedded' Chern-Schwartz-MacPherson treated here
is not an invariant of isomorphism classes.

\begin{ex}\label{exclEx} {\rm
Let $Z_1$ and $Z_2$ be, respectively, a linearly embedded $\P^1$ and a
nonsingular conic in $\P^2$. Denoting by $H$ the hyperplane
class in $\P^2$, we find
$$ 
c(Z_1) = (H+2H^2)\cdot [\P^2]_*   \ \ \ \text{ and } \ \ \ 
c(Z_2) = (2H+2H^2)\cdot [\P^2]_*
$$
while of course $[Z_1]=[Z_2]$ as classes in the Grothendieck group.}
\end{ex}

Thus, in particular, the CSM class $c(X)$ does not define an additive
invariant in the sense of \S \ref{GrDefSec} and does not factor
through the Grothendieck group, as the example above shows.

\smallskip

In certain situations it is however possible to establish a sharp 
correspondence between CSM classes and classes in the Grothendieck
group. For the next result, we adopt the rather unorthodox notation
$H^{-r}$ for the class $[\P^r]_*$ of a linear subspace of a given
projective space. Thus, $1$ stands for the class of a point, $[\P^0]_*$,
and the negative exponents are consistent with the fact that if $H$
denotes the hyperplane class then $H^r \cdot [\P^r]_*= [\P^0]_*$.

\begin{prop}\label{CSMGlinprop}
Let $X$ be a subset of projective space obtained by unions,
intersections, differences of linearly embedded subspaces.
With notation as above, assume
\[
c(X)=\sum a_i H^{-i}\quad.
\]
Then the class of $X$ in the Grothendieck group of varieties
equals
\[
[X]= \sum a_i \bT^i ,
\]
where $\bT=[\bG_m]$ is the class of the multiplicative group, see \S
\ref{MotSect}. 
\end{prop}

Thus, adopting a variable $T=H^{-1}$ in the CSM environment,
and $T=\bT$ in the Grothendieck group environment, the classes
corresponding to subsets as specified in the statement would match
precisely.

\begin{proof}
The formula holds for a linearly embedded $X=\P^r$, since
\[
c(\P^r) = ((1+H)^{r+1}-H^{r+1})\cdot [\P^r]_*
=((1+H)^{r+1}-H^{r+1})\cdot H^{-r}=\frac{(1+H^{-1})^{r+1}-1}{H^{-1}}
\]
and (see \eqref{PrId} below)
\[
[\P^r] = \frac{(1+\bT)^{r+1}-1}{\bT}\quad.
\]
Since embedded CSM classes and classes in the Grothendieck group 
both satisfy inclusion-exclusion, this relation extend to all sets obtained
by ordinary set-theoretic operations performed on linearly
embedded subspaces, and the statement follows.
\end{proof}

Proposition~\ref{CSMGlinprop} applies, for example, to the case of
hyperplane arrangements in $\P^N$: for a hyperplane arrangement,
the information carried by the class in the Grothendieck group of 
varieties is precisely the same as the information carried by the
embedded CSM class. These classes reflect in a subtle way the
combinatorics of the arrangement. 

\smallskip

In a more general setting, it is still possible to enhance the information 
carried by the CSM class in such a way as to establish a tight 
connection between the two environments. For example, CSM 
classes can be treated within a framework with strong similarities
with motivic integration, \cite{Alu1}.

In any case, one should expect that, in many examples, the work 
needed to compute a CSM class should also lead to a computation
of a class in the Grothendieck group. The computations in  \S
\ref{MotSect} and \S \ref{CSMsect} in this paper will confirm 
this expectation for the hypersurfaces corresponding to banana
graphs.

\section{Banana graphs and their motives}\label{MotSect}

In this section we give an explicit formula for the classes 
$[X_{\Gamma_n}]$ of the banana graph hypersurfaces $X_{\Gamma_n}$ 
in the Grothendieck ring. The procedure we adopt to carry out
the computation is the following. We use the Cremona transformation of
\eqref{GraphCrem}. Consider the algebraic simplex $\Sigma_n$ 
placed in the $\P^{n-1}$ on the right-hand-side of the diagram
\eqref{GraphCrem}. The complement of this $\Sigma_n$ in the graph
hypersurface $X_{\Gamma_n}$ is isomorphic to the complement of the 
same union $\Sigma_n$ in the corresponding hyperplane $\cL$ in the
$\P^{n-1}$ on the left-hand-side of \eqref{GraphCrem}, by Lemma
\ref{BananaDual} above. So this provides the easy part of the
computation, and one then has to give explicitly the classes of
the intersections of the two hypersurfaces with the union of the
coordinate hyperplanes. The final formula for the class
$[X_{\Gamma_n}]$ has a simple expression in terms of the classes of
tori $\bT^k$, with $\bT:=[\A^1]-[\A^0]$ the class of the
multiplicative group $\bG_m$. Then $\bT^{n-1}$ is the class of the
complement of $\Sigma_n$ inside $\P^{n-1}$.

In the following we let $1$ denote the class of a point $[\A^0]$. We
use the standard notation $\bL$ for the class $[\A^1]$ of the affine
line (the Lefschetz motive). We also denote, as above, by $\Sigma_n$
the union of coordinate hyperplanes in $\P^{n-1}$ and by $\cS_n$
its singularity locus. 

First notice the following simple identity in the Grothendieck ring.
\begin{equation}\label{PrId}
[\P^r]=\sum_{i=0}^r \bL^r=\frac{1-\bL^{r+1}}
{1-\bL}=\frac{(1+\bT)^{r+1}-1}{\bT}.
\end{equation}
This expression can be thought of as taking place in a localization 
of the Grothendieck ring, but in fact this is not really necessary 
if we take these fractions as just shorthand for their 
unambiguous expansions. 

We introduce the following notation. Suppose given a class $C$ in the
Grothendieck ring which can be written in the form
\begin{equation}\label{classCPn}
 C = a_0 [\P^0] + a_1 [\P^1] + a_2 [\P^2] + \cdots 
\end{equation}
To such a class we assign a polynomial
\begin{equation}\label{fCpolyn}
f(P)=a_0 + a_1 P + a_2 P^2 + \cdots
\end{equation}

\begin{rem}\label{noprodPi} {\rm 
Notice that the formal variable $P$ does not define an element in the
Grothendieck ring, since one sees easily that $P^i P^j\ne
P^{i+j}$. In fact, the variables $P^i$ satisfy a different
multiplication rule, which we denote by $\bullet$ and which is given by
\begin{equation}\label{prodPi}
P^i \bullet P^j = P^{i+j}+P^{i+j-1}+\cdots +P^j-P^{i-1}-\cdots-1  
\end{equation}
and which recovers in this way the class $[\P^i\times
\P^j]$. This follows from Lemma \ref{ProdPi}, by converting each of the
two factors into the corresponding expressions in $\bT$, 
multiplying these as classes in the Grothendieck ring, and then 
converting the result back in terms of the variables $P^i$. }
\end{rem}

\begin{lem}\label{ProdPi}
Let $C$ be a class in the Grothendieck ring that can be written in
terms of classes of projective spaces in the form \eqref{classCPn}.
One can convert it into a function of the class $\bT$ of the form
\begin{equation}\label{Ctorus}
C = \frac{(1+\bT) f(1+\bT) - f(1)}\bT,
\end{equation}
where $f$ is as in \eqref{fCpolyn}.
\end{lem}

\proof One obtains \eqref{Ctorus} from \eqref{classCPn}
using the expression \eqref{PrId} of $[\P^r]$ in terms of
$\bT$. In fact, \eqref{Ctorus} gives the expression of $[\P^r]$ 
as a function of $\bT$ when applied to $f(P)=P^r$. 
\endproof

Conversely, we have a similar way to convert classes in the
Grothendieck ring that can be expressed as a function of the torus
class into a function of the classes of projective spaces.

\begin{lem}\label{gClassT}
Suppose given a class $C$ in the Grothendieck ring that can be
written as a function of the torus class $\bT$, by a polynomial
expression $C=g(\bT)$. Then one obtains an expression of $C$ in terms
of the classes of projective spaces $[\P^r]$ by first taking the
function
\begin{equation}\label{gPfunct}
 \frac{(P-1)g(P-1)+g(-1)}P
\end{equation}
and then replacing $P^r$ by the class $[\P^r]$ 
in the expansion of \eqref{gPfunct} as a polynomial
in the formal variable $P$. 
\end{lem}

\proof The result is obtained by solving for $f$ in \eqref{Ctorus},
which yields the formula \eqref{gPfunct}. 
\endproof

\smallskip

Next we define an operation on classes of the form $C=g(\bT)$, which
one can think of as ``taking a hyperplane section''. Notice that
literally taking a hyperplane section is not a well defined operation 
at the level of the Grothendieck ring, but it does make sense on
classes that are constructed from linearly embedded subspaces of a
projective space, as is the case we are considering.

\begin{lem}\label{hypersec}
The transformation
\begin{equation}\label{HgT}
\cH: g(\bT) \mapsto \frac{g(\bT)-g(-1)}{\bT+1}
\end{equation}
gives an operation on the set of classes in the 
Grothendieck ring that are polynomial functions of the torus class
$\bT$. In terms of classes $[\P^r]$ it corresponds to mapping
$[\P^0]$ to zero and $[\P^r]$ to $[\P^{r-1}]$ for $r\geq 1$.
\end{lem}

\proof One can see that, for $g(\bT)=[\P^r]=\frac{(1+\bT)^{r+1}-1}{\bT}$, 
we have 
$$\frac{g(\bT)-g(-1)}{\bT+1}=\frac{\frac{(1+\bT)^{r+1}-1}\bT-1}
{\bT+1}=\frac{(1+\bT)^r-1}\bT=[\P^{r-1}],$$
or $0$ if $r=0$, so that the operation \eqref{HgT} indeed corresponds
to taking a hyperplane section. The operation is linear in $g$, viewed
as a linear combination of classes of projective spaces, so it
works for arbitrary $g$. 
\endproof

We then have the following preliminary result.

\begin{lem}\label{classSigma}
The class of $\Sigma_{r+1}\subset \P^r$ in the Grothendieck ring is of the form
\begin{equation}\label{Sigmanclass}
[\Sigma_{r+1}]=\frac{(1+\bT)^{r+1}-1-\bT^{r+1}}{\bT}
=\sum_{i=1}^r \binom{r+1}i \bT^{r-i}.
\end{equation}
Intersecting with a transversal hyperplane $\cL$ then gives
\begin{equation}\label{SigmaLclass}
[\cL\cap \Sigma_{r+1}] = \frac{(1+\bT)^r-1}\bT - \bT^{r-1}+\bT^{r-2}-\bT^{r-3}+
\cdots \pm 1 .
\end{equation}
\end{lem}

\proof The class of the complement of $\Sigma_{r+1}$ in $\P^r$ is the
torus class $\bT^r$. In fact, the complement of $\Sigma_{r+1}$ consists 
of all $(r+1)$-tuples $(1{:}*{:}\cdots{:}*)$, where each $*$ is a
nonzero element of the ground field. It then follows directly that
the class of $\Sigma_{r+1}$ has the form \eqref{Sigmanclass}, using
the expression \eqref{PrId} for the class $[\P^r]$. One then
applies the transformation $\cH$ of \eqref{HgT} to obtain
$$ \begin{array}{ll}
[\cL\cap \Sigma_{r+1}] &=\left(\frac{(1+\bT)^{r+1}-1-\bT^{r+1}}\bT
-\frac{-1-(-1)^{r+1}}{-1}\right)/(\bT+1) \\[3mm]
&=\frac{(1+\bT)^r-1}\bT - \frac{\bT^r-(-1)^r}{\bT+1}\end{array} $$
from which \eqref{SigmaLclass} follows. 
\endproof

\begin{center}
\begin{figure}
\includegraphics[scale=0.4]{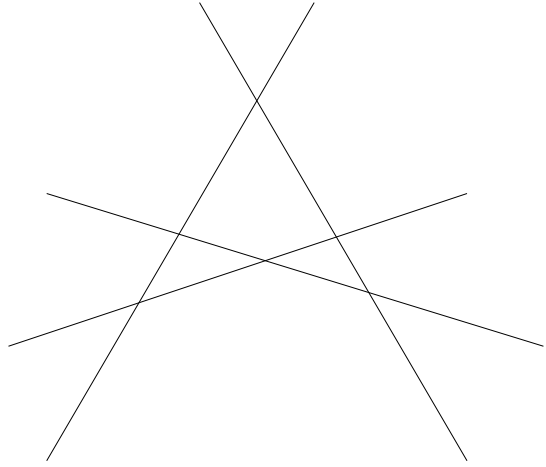}
\caption{The trace $\Sigma'_4\subset \P^2$ of the algebraic simplex
$\Sigma_4\subset \P^3$
\label{FigFourlines}}
\end{figure}
\end{center}

\begin{defn}\label{SigmaTr}
The trace $\Sigma'_{r+1} \subset \P^{r-1}$ of the algebraic simplex
$\Sigma_{r+1}\subset \P^r$ is the intersection of $\Sigma_{r+1}$ 
with a general hyperplane. It is a union of $r+1$ hyperplanes in
$\P^{r-1}$ meeting with normal crossings. 
\end{defn}

For instance, $\Sigma'_4$ consists of the transversal union of four
lines as in Figure \ref{FigFourlines} and by \eqref{SigmaLclass} its class is 
$$ [ \Sigma'_4 ]=  \frac{(1+\bT)^3-1}\bT - \bT^2+\bT-1=4\bT+2 .$$

The first part of the computation of the class of the graph
hypersurface $X_{\Gamma_n}$ for the banana graph $\Gamma_n$ is then
given by the following result.

\begin{prop}\label{classXGamman}
Let $X_{\Gamma_n}\subset \P^{n-1}$ be the hypersurface of the
$n$-th banana graph $\Gamma_n$. Then 
\begin{equation}\label{SigmacomplXGamma}
[X_{\Gamma_n}\smallsetminus \Sigma_n]=\bT^{n-2}-\bT^{n-3}+\bT^{n-4}-\cdots
+(-1)^n .
\end{equation}
\end{prop}

\proof We know by Lemma \ref{BananaDual} that
$X_{\Gamma_n}\smallsetminus \Sigma_n \cong \cL \smallsetminus
\Sigma_n$ via the Cremona transformation, with $\cL=\P^{n-2}$ the hyperplane 
\eqref{Lhyp}. This hyperplane intersects $\Sigma_n$ transversely, so
that \eqref{SigmaLclass} applies and gives
$$[\cL\smallsetminus \Sigma_n]=[\cL]-[\cL\cap \Sigma_n]=\frac{\bT^{n-1}-(-1)^{n-1}}
{\bT+1}.$$
\endproof

Next we examine how the graph hypersurface $X_{\Gamma_n}$ intersects
the algebraic simplex $\Sigma_n$. 

\begin{lem}\label{XGammaS}
The graph hypersurface $X_{\Gamma_n}$ intersects the algebraic simplex
$\Sigma_n \subset \P^{n-1}$ along the singularity subscheme $\cS_n$ of
$\Sigma_n$. 
\end{lem}

\proof One can see this directly by comparing the defining equation
\eqref{PolyBn} of $X_{\Gamma_n}$ with the ideal $\cI_{\cS_n}$ of
\eqref{idealS} of the singularity subscheme $\cS_n$ of
$\Sigma_n$. 
\endproof

The class in the Grothendieck ring of the singular locus $\cS_n$ of
$\Sigma_n$ is given by the following result.

\begin{lem}\label{classSnLem}
The class of $\cS_{r+1}\subset \P^r$ is given by
\begin{equation}\label{classSn}
[\cS_{r+1}] = [\Sigma_{r+1}]- (r+1)\bT^{r-1}
=\sum_{i=2}^r \binom{r+1}i \bT^{r-i} .
\end{equation}
\end{lem}

\proof Each coordinate hyperplane $\P^{r-1}$ in $\Sigma_{r+1}\subset
\P^r$ intersects the others along its own algebraic simplex
$\Sigma_r$. Thus, to obtain the class of $\cS_{r+1}$ from the class
of $\Sigma_{r+1}$ in the Grothendieck ring we just need to subtract 
the class of the $r+1$ complements of $\Sigma_r$ in the $r+1$
components of $\Sigma_{r+1}$. We then have
$$ [\cS_{r+1}] = [\Sigma_{r+1}]- (r+1)\bT^{r-1}
=\frac{(1+\bT)^{r+1}-1-(r+1)\bT^r-\bT^{r+1}}\bT. $$
This gives the formula \eqref{classSn}. 
\endproof

We then have the following result. 

\begin{thm}\label{classXGammaGr}
The class in the Grothendieck ring of the graph hypersurface
$X_{\Gamma_n}$ of the banana graph $\Gamma_n$ is given by
\begin{equation}\label{clXGammaGr}
[X_{\Gamma_n}]=\frac{(1+\bT)^n-1}\bT-\frac{\bT^n-(-1)^n}{\bT+1}-n\,
\bT^{n-2} .
\end{equation}
\end{thm}

\proof We write the class $[X_{\Gamma_n}]$ in the form
$$ [X_{\Gamma_n}] =[ X_{\Gamma_n} \smallsetminus \Sigma_n ]+[\cS_n]. $$
Using the results of Lemma \ref{classSnLem} and Proposition
\ref{classXGamman} we write this as
$$ = \frac{\bT^{n-1}-(-1)^{n-1}}{\bT+1}+\frac{(1+\bT)^n-1-n\bT^{n-1}
-\bT^n}\bT ,$$
from which \eqref{clXGammaGr} follows.
\endproof

The formula \eqref{clXGammaGr} expresses the class $[X_{\Gamma_n}]$ as
$$ [X_{\Gamma_n}] = [\Sigma_n']-n\bT^{n-2},  $$ 
\ie as the class of the union $\Sigma_n'$ of $n$ hyperplanes meeting
with normal crossings (as in Definition \ref{SigmaTr}), corrected by 
$n$ times the class of an $n-2$-dimensional torus.

\begin{ex}\label{exP2} {\rm 
In the case $n=3$ of Figure \ref{BananaFig}, \eqref{clXGammaGr} shows 
that the class of the hypersurface $X_{\Gamma_3}\subset \P^2$ is equal
to the class of the union of four transversal lines, minus three times a
$1$-dimensional torus, \ie that we have
$$ [X_{\Gamma_3}]=4\bT+2-3\bT=\bT+2=[\P^1]. $$
This can also be seen directly from the fact that the equation
$$ \Psi_{\Gamma_3}=t_1t_2+t_2t_3+t_1t_3=0 $$
defines a nonsingular conic in the plane.}
\end{ex}

\begin{ex}\label{exP3}{\rm  In the case $n=4$ of Figure
\ref{BananaFig}, the hypersurface $X_{\Gamma_4}$ is a cubic surface in
$\P^3$ with four singular points. The class in the Grothendieck ring
is 
$$ [X_{\Gamma_4}]= \bT^2 + 5\bT + 5 . $$
}\end{ex}

In terms of the Lefschetz motive $\bL$, the formula \eqref{clXGammaGr}
reads equivalently as
\begin{equation}\label{clXGammaGrL}
[X_{\Gamma_n}]= \frac{\bL^n-1}{\bL-1}-\frac{(\bL-1)^n-(-1)^n}{\bL}-n\,
(\bL-1)^{n-2} .
\end{equation}

In the context of parametric Feynman integrals, it is the complement
of the graph hypersurface in $\P^{n-1}$ that supports the period
computed by the Feynman integral. Thus, in general, one is interested 
in the explicit expression for the motive of the complement. It so 
happens that in the particular case of the banana graphs the
expression for the class of the hypersurface complement is in fact
simpler than that of the hypersurface itself. 

\begin{cor}\label{classXGammaCompl}
The class of the hypersurface complement $\P^{n-1}\smallsetminus
X_{\Gamma_n}$ is given by
\begin{equation}\label{classCompl}
\begin{array}{rl}
[\P^{n-1}\smallsetminus X_{\Gamma_n}]= &
\frac{\bT^n-(-1)^n}{\bT+1}+n\, \bT^{n-2} \\[3mm]
=& \bT^{n-1}+(n-1) \bT^{n-2}+\bT^{n-3}-\bT^{n-4}+\bT^{n-5}+\cdots \pm
1. \end{array}
\end{equation}
\end{cor}

\proof By \eqref{PrId} we see that the first term in
\eqref{clXGammaGr} is in fact the class $[\P^{n-1}]$, hence the class
$[\P^{n-1}\smallsetminus X_{\Gamma_n}]=[\P^{n-1}]-[X_{\Gamma_n}]$ is
given by \eqref{classCompl}
\endproof

\begin{cor}\label{EulCharGr}
The Euler characteristic of $X_{\Gamma_n}$ is given by the formula
\eqref{chiBn}.
\end{cor}

\proof The Euler characteristic is an additive invariant, hence it
determines a ring homomorphism from the Grothendieck ring of varieties
to the integers. Moreover, tori have zero Euler characteristic, so
that $\chi(\bT^r)=0$ for all $r\geq 1$. Then the formula
\eqref{classCompl} for the class of the hypersurface complement shows
that 
$$ \chi(\P^{n-1}\smallsetminus X_{\Gamma_n})= \chi(\bT^{n-1})+(n-1)
\chi(\bT^{n-2})+\chi(\bT^{n-3})-\cdots \pm 1 = (-1)^{n-1}. $$
Since $\chi(\P^{n-1})=n$ we obtain 
$$ \chi(X_{\Gamma_n})= \chi(\P^{n-1}) -  \chi(\P^{n-1}\smallsetminus
X_{\Gamma_n}) = n + (-1)^n $$
as in \eqref{chiBn}. 
\endproof

In \S \ref{CSMsect} below, we derive the same Euler characteristic
formula in a different way, from the calculation of the CSM class of
$X_{\Gamma_n}$. 

\begin{rem}\label{Rat}{\em
Notice that, if we expand in \eqref{clXGammaGr} the first term in the
form $[\P^{n-1}]=\bT^{n-1}+n\bT^{n-2}+\dots$, we see that the dominant
term in $[X_{\Gamma_n}]$ is $\bT^{n-2}$. This is not surprising, since
for the banana graphs the hypersurfaces $X_{\Gamma_n}$ are rational. 
}\end{rem}

\begin{rem}\label{missingPn} {\em
The previous remark explains the appearance of a term $n \bT^{n-2}$ in
the expression \eqref{classCompl}. The remaining terms are an
alternating sum of tori. This term can be viewed as
\begin{equation}\label{altTsect}
 \frac{\bT^n-(-1)^n}{\bT+1} = \frac{g(\bT)-g(-1)}{\bT+1}, 
\end{equation}
for $g(\bT)=\bT^n$. According to Lemma \ref{hypersec}, this is 
the class of the hyperplane section of the complement of the algebraic 
simplex $\Sigma_{n+1}$ in $\P^n$. However, how geometrically one can
associate a $\P^n$ to a graph hypersurface $X_{\Gamma_n}\subset
\P^{n-1}$ is unclear, so that a satisfactory conceptual explanation of
the occurrence of \eqref{altTsect} in \eqref{classCompl} is still
missing. }\end{rem}

For completeness we also give the explicit formula of the class
\eqref{classCompl} written in terms of classes $[\P^r]$. 

\begin{cor}\label{classComplPrCor}
In terms of classes of projective spaces the class 
$[\P^{n-1}\smallsetminus X_{\Gamma_n}]$ is given by
\begin{equation}\label{classComplPr}
[\P^{n-1}\smallsetminus X_{\Gamma_n}] =
\sum_{k=0}^{n-1}\binom{n+1}{k+2} (-1)^{n-1-k}\, [\P^k] +
n\,\sum_{k=0}^{n-2}\binom{n-1}{k+1} (-1)^{n-2-k}\, [\P^k]. 
\end{equation}
\end{cor}

\proof The formula \eqref{classComplPr} is
obtained easily using the transformation rules of Lemma \ref{gClassT}
to go from expressions 
in $\bT$ to expressions in $[\P^r]$, so that
$$ \begin{array}{rl}
(\bT^n-(-1)^n)/(\bT+1) \mapsto &\left((P-1)\frac{(P-1)^n-(-1)^n}P + n(-1)^{n-1}\right)/P\\[2mm]
= &\left((P-1)^{n+1}-(n+1)(-1)^nP-(-1)^{n+1}\right)/P^2\\[2mm]
= &\sum_{k=0}^{n-1}\binom{n+1}{k+2} (-1)^{n-1-k}\, [\P^k],
\end{array} $$
and for the second term
$$ \begin{array}{rl}
n\,\bT^{n-2} \mapsto
&n\,\left((P-1)(P-1)^{n-2} + (-1)^{n-2}\right)/P\\[2mm]
=&n\,\left((P-1)^{n-1}-(-1)^{n-1}\right)/P\\[2mm]
=&n\,\sum_{k=0}^{n-2}\binom{n-1}{k+1} (-1)^{n-2-k}\, [\P^k].
\end{array} $$
\endproof

\medskip

\subsection{Classes of dual graphs}\label{DualClassSec}

In the result obtained above, we used essentially the relation between
the graph hypersurface $X_{\Gamma_n}$ and the hypersurface of the dual
graph, which is, in this case, a hyperplane. More generally, although
one cannot obtain an explicit formula, one can observe that for any
given planar graph the relation between the hypersurface $X_{\Gamma}$ 
of the graph and that of the dual graph $X_{\Gamma^\vee}$ gives a
relation between the classes in the Grothendieck ring, which can be
expressed as follows.

\begin{prop}\label{DualGraphXclass}
Let $\Gamma$ be a planar graph with $n=\#E(\Gamma)$ and let
$\Gamma^\vee$ be a dual graph. Then the classes in the Grothendieck
ring satisfy 
\begin{equation}\label{cldiffDualGr}
[X_\Gamma]-[X_{\Gamma^\vee}]=[\Sigma_n\cap X_\Gamma] -[\Sigma_n \cap
X_{\Gamma^\vee}]. 
\end{equation}
\end{prop}

\proof The result is a direct consequence of Corollary
\ref{DualGrCoroll}. 
\endproof

\section{CSM classes for banana graphs}\label{CSMsect}

We now give an explicit formula for the Chern--Schwartz--MacPherson
class of the hypersurfaces of the banana graphs, for an arbitrary number
of edges. 

The computation of the CSM class is substantially more involved than
the computation of the class in the Grothendieck ring we obtained in
the previous section, although the two carry strong formal
similarities, due to the fact that both are based on a similar
inclusion--exclusion principle. In fact, the information carried by
the CSM class is more refined than the decomposition in the
Grothendieck ring of varieties, as it captures more sophisticated
information on how the building blocks are embedded in the ambient
space. This will be illustrated rather clearly by our explicit
computations. In particular, the explicit formula for the CSM class
uses in an essential way a special formula for pullbacks of
differential forms with logarithmic poles.

In order to avoid any possible confusion between homology classes and
classes in the Grothendieck ring (even though the context should
suffice to distinguish them), we use here as in \S \ref{CSMdefSec}
the notation $[\P^r]_*$ for homology classes or classes in the Chow
group (in an ambient $\P^{n-1}$), while reserving the symbol $[\P^r]$
for the class in the Grothendieck ring, as already used in \S
\ref{MotSect} above. The homology class $[\P^r]_*$ can be expressed in
terms of the hyperplane class $H$ and the ambient $\P^{n-1}$ 
as $[\P^r]_* = H^{n-1-r} [\P^{n-1}]_*$.

\subsection{Characteristic classes of blowups}

Let $D$ be a divisor with simple normal
crossings and nonsingular components $D_i$, $i=1,\dots,r$, in a 
nonsingular variety~$M$. Then $TM(-\log(D))$ denotes the sheaf of
vector fields with logarithmic zeros (\ie the dual of the
sheaf $\Omega^1_M(\log D)$ of 1-forms with logarithmic poles). 
In terms of Chern classes one has (\cf \eg \cite{Alu1}) 
$$ c(TM(-\log D))=\frac{c(TM)}{(1+D_1)\cdots (1+D_r)} . $$
This formula has useful applications in the calculation of 
CSM classes, especially because it behaves nicely under pushforwards
as shown in \cite{Alu1} and \cite{Alu2}. What we need here is a more
surprising pullback formula, which can be stated as follows.

\begin{thm}\label{pbForms}
Let $\pi: W \to V$ be the blowup of a nonsingular variety~$V$
along a nonsingular subvariety $B$, with exceptional divisor~$F$.
Let $E_j$, $j\in J$, be nonsingular irreducible hypersurfaces of $V$, 
meeting with normal crossings. Assume that $B$ is cut out by some 
of the $E_j$'s. Denote by $F_j$ the proper transform of $E_j$ in $W$.
Then the sheaf of 1-forms with logarithmic poles along $E$
is preserved by the pullback, namely
\begin{equation}\label{bundlepb}
\Omega^1_M(\log(F+\sum F_j))=\pi^* \Omega^1_M(-\log (\sum E_j)).
\end{equation}
\end{thm}

\proof There is an inclusion $\pi^* \Omega^1_M(\log (\sum E_j))\subset 
\Omega^1_M(\log(F+\sum F_j))$, and it suffices to see that this is an
equality locally analytically. To this purpose, choose local
coordinates $x_1,\ldots,x_n$ in $V$, so that $E_j$ is given 
by $(x_j)$, $j=1,\dots,k$. Assume $B$ has ideal $(x_1, \dots,x_\ell)$,
and choose local parameters $y_1,\dots,y_n$ in $W$ so that $\pi$ is
expressed by 
$$\left\{\aligned
x_1 &= y_1 \\
x_2 &= y_1 y_2 \\
\cdots \\
x_\ell &= y_1 y_\ell \\
x_{\ell+1} &= y_{\ell+1} \\
\cdots \\
x_n &= y_n
\endaligned\right.$$
Then local sections of $\pi^*\Omega^1(\log (\sum E_j))$ are spanned by
$$\frac{dy_1}{y_1}\quad,\quad 
\frac{dy_1}{y_1}+\frac{dy_2}{y_2}\quad,\quad 
\cdots\quad,\quad
\frac{dy_1}{y_1}+\frac{dy_\ell}{y_\ell}\quad,\quad 
\frac{dy_{\ell+1}}{y_{\ell+1}}\quad,\quad 
\cdots\quad,\quad
\frac{dy_n}{y_n}$$
These clearly span the whole of $\Omega^1(\log(F+\sum F_j))$, as
claimed. Thus, we obtain \eqref{bundlepb}.
\endproof

One derives directly from this result the following formula for Chern
classes. 

\begin{cor}\label{prespb}
Under the same hypothesis as Theorem \ref{pbForms}, the Chern classes
satisfy 
\begin{equation}\label{pbformula}
\frac{c(TW)}{(1+F)\prod_{j\in J}(1+F_j)} \cap [W]
=\pi^*\left(\frac{c(TV)}{\prod_{j\in J}(1+E_j)}\cap [V]\right) .
\end{equation}
\end{cor}

In other words, if $B$ is cut out by a selection of the components
$E_j$, then the pullback of the total Chern class of the bundle
of vector fields with logarithmic zeros along $E$ equals the one 
of the analogous bundle upstairs.

The main consequence of Theorem \ref{pbForms} and Corollary
\ref{prespb} which is relevant to the case of graph hypersurfaces 
is given by the following application.

\begin{defn}\label{defadapt}
Let $V$ be a nonsingular variety, and let
$E_j$, $j\in J$, be nonsingular divisors meeting with normal crossings 
in $V$. A proper birational map $\pi: W \to V$ is a blowup 
{\em adapted\/} to the divisor with normal crossings if it is the blowup 
of $V$ along a subvariety $B\subset V$
cut out by some of the $E_j$'s.
\end{defn}

Notice that $W$ carries a natural divisor
with normal crossings, that is, the union of the exceptional divisor $F$
and of the proper transforms $F_j$ of the divisors $E_j$. The blowup
maps the complement of $W$ to this divisor isomorphically to the
complement in $V$ of the divisor $\cup E_j$. It makes sense then to talk
about a {\em sequence of adapted blowups,} by which we mean that each 
blowup in the sequence is adapted to the corresponding normal crossing
divisor. We then have the following consequence of Theorem
\ref{pbForms} and Corollary \ref{prespb} above.

\begin{cor}\label{adapt}
Let $V$ be a nonsingular variety, and $E_j$ be nonsingular divisors
meeting with normal crossings in $V$. Let $U$ denote the complement
of the union $E=\cup_{j\in J} E_j$.
Let $\pi: W \to V$ be a proper birational map dominated by a 
sequence of adapted blowups. In particular, $\pi$ maps $\pi^{-1}(U)$
isomorphically to $U$. Then
\begin{equation}\label{CSMpb}
c(\one_{\pi^{-1}(U)})=\pi^* c(\one_U) .
\end{equation}
\end{cor}

\begin{proof}
Let $\tilde\pi: \tilde V \to V$ be a sequence of adapted morphisms
dominating $\pi$:
$$\xymatrix{
\tilde V \ar[r]^\alpha \ar[dr]_{\tilde\pi} & W \ar[d]^\pi \\
& V
}$$
The divisor $\tilde E=\tilde\pi^{-1}(\cup E_j)$ is then a divisor with normal
crossings, and $\tilde\pi^{-1}(U)$ is its complement in $\tilde V$.
By Corollary~\ref{prespb}, we have the identity
$$c(T\tilde V(-\log\tilde E))\cap [\tilde V]=\tilde\pi^*(c(TV(-\log E))\cap [V])
=\alpha^* \pi^* (c(TV(-\log E))\cap [V]) .$$
As in \eqref{nuiTW} of \S \ref{CSMvsGrSec}, this is saying that
$$ c(\one_{\tilde\pi^{-1}(U)})=\alpha^* \pi^*( c(\one_U)). $$
The statement then follows by pushing forward through $\alpha$ (applying
MacPherson's theorem), since $\alpha_*\alpha^*=1$ as $\alpha$ is
proper and birational. 
\end{proof}

The identity of CSM classes happens in the homology (or Chow group) of
$W$. Notice that we are not assuming here that $W$ is nonsingular. One
also has $\pi_* c(\one_{\pi^{-1}(U)})=c(U)$ in the homology (Chow
group) of $V$, by MacPherson's theorem \cite{MacPher} on functoriality
of CSM classes. What is surprising about \eqref{CSMpb} is that for
this class of morphisms one can do for pullbacks what functoriality
usually does for pushforward. 

\subsection{Computing the characteristic classes}

In this section we give the explicit formula for the CSM class of the
graph hypersurface $X_{\Gamma_n}$ of the banana graph $\Gamma_n$. The
procedure is somewhat similar conceptually to the one we used in the
computation of the class in the Grothendieck ring, namely we will
use the inclusion--exclusion property of the Chern class and separate 
out the contributions of the part of $X_{\Gamma_n}$ that lies in the
complement of the algebraic simplex $\Sigma_n \subset \P^{n-1}$ 
and of the intersection $X_{\Gamma_n}\cap \Sigma_n$, using the Cremona
transformation to compute the contribution of the first and
inclusion-exclusion again to compute the class of the latter.

As above, let $\cS_n$ be the singularity subscheme of
$\Sigma_n$. We begin by the following preliminary result.

\begin{prop}\label{CSMsingSigma}
The CSM class of $\cS_n$ is given by 
\begin{equation}\label{CSMofSn}
c(\cS_n)=((1+H)^n-1-nH-H^n)\cdot [\P^{n-1}]_*.
\end{equation}
\end{prop}

\proof Since $\cS_n$ is defined by the ideal \eqref{idealS} 
of the codimension two intersections of the coordinate planes of
$\P^{n-1}$, one can use the inclusion-exclusion property to compute
\eqref{CSMofSn}. Equivalently, one can use the result of
\cite{Alu06a}, which shows that, for a locus that is a union of toric 
orbits, the CSM class is a sum of the homology classes of the
orbit closures. Thus, one can write the CSM class of $\P^{n-1}$ as 
the sum of the CSM class of $\cS_n$, the homology classes of the
coordinate hyperplanes, and the homology class of the whole
$\P^{n-1}$, \ie
$$ c(\P^{n-1})= (c(\cS_n) + nH + 1)\cdot [\P^{n-1}]_*, $$
where the two latter terms correspond to the classes of the 
closures of the higher dimensional orbits. Since
$c(\P^{n-1})=((1+H)^n-H^n)\cdot [\P^{n-1}]_*$, this gives \eqref{CSMofSn}.
\endproof

We now concentrate on the complement $X_{\Gamma_n}\smallsetminus
\Sigma_n$. We again use Lemma \ref{BananaDual} to describe this, via
the Cremona transformation, in terms of $\cL \smallsetminus
\Sigma_n$, with $\cL$ the hyperplane \eqref{Lhyp}. We have the
following result.

\begin{lem}\label{CSMLSigma}
Let $\pi_1: \cG(\cC)\to \P^{n-1}$ be as in \eqref{GraphCrem}. Then 
\begin{equation}\label{CSMpi1L}
c(\pi_1^{-1}(\cL\smallsetminus \Sigma_n))=\pi_1^*(c(\cL\smallsetminus
\Sigma_n)). 
\end{equation}
\end{lem}

\proof By Corollary~\ref{adapt}, it suffices to show that the restriction of $\pi_1$
to $\pi^{-1}(\cL)$ is {\em adapted} to $\cL\cap \Sigma_n$.
By (2) and (3) of Lemma \ref{SandL}, we know that $\pi_1^{-1}(\cL)$ is
the blowup of $\cL$ along $\cL\cap \cS_n$, that is, the singularity subscheme
of $\cL\cap \Sigma_n$. The blowup of a variety along the singularity
subscheme of a divisor with simple normal crossings is
dominated by the sequence of blowups along the intersections of the
components of the divisor, in increasing order of dimension. This sequence
is adapted, hence the claim follows. Equivalently, notice that $\pi_1: \cG(\cC) \to
\P^{n-1}$ is itself dominated by a sequence adapted to
$\Sigma_n$. Moreover, $\cL$ and its proper transform intersect
all centers of the blowups in the sequence transversely. This also
shows that the restriction of $\pi_1$ to $\pi_1^{-1}(\cL)$
is adapted to $\cL\cap \Sigma_n$.
\endproof

We have the following result for the CSM class of
$\cL\smallsetminus \Sigma_n$ in terms of the homology (Chow group) 
classes $[\P^r]_*$.

\begin{lem}\label{CSMLhLem}
The CSM class of $\cL\smallsetminus \Sigma_n$ is given by
\begin{equation}\label{CSMLSigmaCompl}
c(\cL\smallsetminus \Sigma_n)=[\P^{n-2}]_* - [\P^{n-3}]_* + \cdots + (-1)^n
[\P^0]_* .
\end{equation} 
Let $h$ denote the homology class of the hyperplane section in the
source $\P^{n-1}$ of diagram \eqref{GraphCrem}. Then
\eqref{CSMLSigmaCompl} is written equivalently as
\begin{equation}\label{CSMLSigmaCompl2}
c(\cL\smallsetminus \Sigma_n)= (1+h)^{-1} h \cdot [\P^{n-1}]_* .
\end{equation}
\end{lem}

\proof The divisor $\Sigma_n$ has $n$ components with homology class $h$, hence so does
$\cL\cap \Sigma_n$. Since the CSM class of a divisor with normal
crossings is computed by the Chern class of the bundle of vector fields with 
logarithmic zeros along the components of the divisor, we find
$$ c(\cL\smallsetminus \Sigma_n)= \frac{c(T\cL)\cap [\cL]_*}{(1+h)^n}
=\frac{(1+h)^{n-1}}{(1+h)^n}\,h\cdot [\P^{n-1}]_* . $$
\endproof

We then have the following result that gives the formula for
$c(X_{\Gamma_n})$. 

\begin{thm}\label{CSMXGammaThm}
The (push-forward to $\P^{n-1}$ of the) CSM class of the banana 
graph hypersurface $X_{\Gamma_n}$ is given by 
\begin{equation}\label{CSMXGamma}
c(X_{\Gamma_n})=((1+H)^n-(1-H)^{n-1}-nH-H^n)\cdot[\P^{n-1}]_* .
\end{equation}
\end{thm}

\proof Using inclusion-exclusion for CSM classes we have
$$ c(X_{\Gamma_n}) =c(X_{\Gamma_n}\cap \Sigma_n) +
c(X_{\Gamma_n}\smallsetminus \Sigma_n). $$
By Lemma \ref{XGammaS}, we know that $X_{\Gamma_n}\cap
\Sigma_n=\cS_n$, hence the first term is given by
\eqref{CSMofSn}. Thus, we are reduced to showing that
\begin{equation}\label{CSMSigmaCompl}
c(X_{\Gamma_n}\smallsetminus \Sigma_n)=(1-(1-H)^{n-1})\cdot
[\P^{n-1}]_* .
\end{equation}
Combining Lemmata \ref{CSMLSigma} and \ref{CSMLhLem}, we find
$$ \begin{array}{rl}
c(X_{\Gamma_n}\smallsetminus \Sigma_n) = & {\pi_2}_* \pi_1^* \left(\sum_{i=1}^{n-1} (-1)^{i-1}
h^i \cdot [\P^{n-1}]_* \right) \\[2mm]
=& {\pi_2}_* \left(\sum_{i=1}^{n-1} (-1)^{i-1}h^i \cdot
[\cG(\cC)]_* \right),
\end{array} $$
where we view $h$ as a divisor class on $\cG(\cC)$, suppressing the
pullback notation. Let $H$ denote the hyperplane class in the target
$\P^{n-1}$ of diagram \eqref{GraphCrem}, as well as its pullback to
$\cG(\cC)$. Notice that, by \eqref{eqgraphCrem}, $\cG(\cC)$ is a complete
intersection of $n-1$ hypersurfaces of homology class $h+H$ in
$\P^{n-1}\times \P^{n-1}$. Thus, we obtain
$$ c(X_{\Gamma_n}\smallsetminus \Sigma_n)={\pi_2}_* 
\left(\sum_{i=1}^{n-1} (-1)^{i-1}h^i (h+H)^{n-1}\cdot 
[\P^{n-1}\times \P^{n-1}]_* \right). $$
Finally, we have to evaluate the pushforward via $\pi_2$. We can write
$$ c(X_{\Gamma_n}\smallsetminus \Sigma_n) = \sum_{i=1}^{n-1} a_i
H^i\cdot [\P^{n-1}]_* , $$
where we need to evaluate the integers $a_i$. Since
$$ a_i = \int H^{n-1-i}\cdot c(X_{\Gamma_n}\smallsetminus \Sigma_n),$$
by the projection formula we obtain
$$ a_i = \int H^{n-1-i}\sum_{i=1}^{n-1} (-1)^{i-1}h^i (h+H)^{n-1}
\cdot [\P^{n-1}\times \P^{n-1}]_*. $$
In $\P^{n-1}\times \P^{n-1}$, the only nonzero monomial in $h$, $H$
of degree $2n-2$ is $h^{n-1}H^{n-1}$, which evaluates
to~$1$. Therefore, we have
$$ a_i=(-1)^{i-1}\cdot \text{coefficient of $h^{n-1-i}H^i$ in $(h+H)^{n-1}$}
=(-1)^{i-1}\binom {n-1}{i} .$$
We then obtain
$$ \begin{array}{rl}
c(X_{\Gamma_n}\smallsetminus \Sigma_n) = & \sum_{i=1}^{n-1} a_i
H^i\cdot [\P^{n-1}]_* \\[2mm]
= & \sum_{i=1}^{n-1} (-1)^{i-1} \binom{n-1}i H^i \cdot
[\P^{n-1}]_* \\[2mm]
= & (1-(1-H)^{n-1})\cdot [\P^{n-1}]_* .
\end{array} $$
\endproof

This gives a different way of computing the topological Euler
characteristic of $X_{\Gamma_n}$, which we already derived from the
class in the Grothendieck ring in Corollary \ref{EulCharGr}.

\begin{cor}\label{CorEulChar}
The Euler characteristic of the banana graph hypersurface
$X_{\Gamma_n}$ is given by the formula \eqref{chiBn}.
\end{cor}

\proof For a projective hypersurface $X$ the value of the Euler
characteristic $\chi(X)$ can be read off the CSM class as the
coefficient of the top degree term. Thus, from  \eqref{CSMXGamma}
we obtain $\chi(X_{\Gamma_n})=n+ (-1)^n$.
\endproof

\begin{rem}{\em
The coefficient of $H^k$ in the CSM class is as prescribed in 
\eqref{CSMguess}. In particular, these coefficients are positive 
for all $n\ge 2$ and $1\le k\le n-1$. Thus, banana graphs provide
an infinite family of graphs for which Conjecture~\ref{posconj} holds.
}\end{rem}

\begin{rem}{\em
As pointed out in \S\ref{CSMvsGrSec}, CSM classes are defined (as classes
in the Chow group of an ambient variety) for locally closed subsets. 
It follows from
Theorem~\ref{CSMXGammaThm} that the CSM class of the
complement of $X_{\Gamma_n}$ in $\P^{n-1}$ is
\[
c(\P^{n-1}\smallsetminus X_{\Gamma_n})= ((1-H)^{n-1}+nH)
\cap [\P^{n-1}]_*.
\]
}\end{rem}

\subsection{The CSM class and the class in the Grothendieck ring}

We discuss here the formal similarity, as well as the discrepancy,
between the expression for the CSM class and the formula for the
class in the Grothendieck ring of the graph hypersurface
$X_{\Gamma_n}$. 

As noted in Propostion~\ref{CSMGlinprop}, the CSM class and the
class in the Grothendieck group carry the same information for
subsets of projective space consisting of unions of linear subspaces.
The algebraic simplex, as well as its trace on a transversal 
hyperplane, are subsets of this type. Thus, some of the work
performed in \S\S\ref{MotSect} and~\ref{CSMsect} is redundant.

The class $[X_{\Gamma_n}]$ of the graph
hypersurface $X_{\Gamma_n}$ of the banana graph $\Gamma_n$ can be
separated into two parts, only one of which --the part that comes
from the simplex-- is linearly embedded. These two parts are
responsible, respectively, for the formal similarity and for the
discrepancy between the expression for the class $[X_{\Gamma_n}]$ and
the one for $c(X_{\Gamma_n})$. 

In fact, denoting the unorthodox $H^{-1}$ of Proposition~\ref{CSMGlinprop}
by a variable~$T$, the CSM class of $X_{\Gamma_n}$ has the form
\begin{equation}\label{CSMvarT}
c(X_{\Gamma_n}) = \frac{(1+T)^n-1}T-(T-1)^{n-1}-n\, T^{n-2}.
\end{equation}
The central term is the one that differs from the expression
\eqref{clXGammaGr}.
Adopting the same variable $T$ for the class $\bT$ of a torus, 
the discrepancy is measured by the amount
$$ \frac{T^n-(-1)^n}{T+1}-(T-1)^{n-1} .$$

\section{Classes of cones}\label{ConesSect}

We make here a general observation which may be useful in other
computations of CSM classes and classes in the Grothendieck ring 
for graph hypersurfaces. One can observe
that often the graph hypersurfaces $X_\Gamma$ happen to be cones over
hypersurfaces in smaller projective spaces. 

There are simple operations one can perform on a given graph, which
ensure that the resulting graph will correspond to a hypersurface that
is a cone. Here is a list:
\begin{itemize}
\item Subdividing an edge.
\item Connecting two graphs by a pair of edges. 
\item Appending a tree to a vertex (in this case the resulting graph
will not be 1PI).
\end{itemize}
One can see easily that in each of these cases the resulting
hypersurface is a cone, since in the first two cases 
the resulting graph polynomial $\Psi_\Gamma$ will depend on two 
of the variables only through their linear combination
$t_i+t_j$, while in the last case $\Psi_\Gamma$ does not depend 
on the variables of the edges in the tree. 

It may then be useful to provide an explicit formula for computing the
CSM class and the class in the Grothendieck ring for cones. 
The result can be seen as a generalization of the simple formula for
the Euler characteristic. 

\begin{lem}\label{EulCharCone}
Let $C^k(X)$ be a cone in $\P^{n+k}$ of a hypersurface $X\subset
\P^n$. Then the Euler characteristic satisfies 
$$ \chi(C^k(X))=\chi(X) + k. $$
\end{lem}

\proof Consider first the case of $C(X)=C^1(X)$. We have
$$ C(X)=(X\times \P^1) / (X\times \{ pt \}) = (X\times \A^1) \cup \{ pt \}, $$
from which, by the inclusion-exclusion property of the Euler
characteristic we immediately obtain $\chi(C(X))=\chi(X) + 1$. The
result then follows inductively.
\endproof

The case of the CSM class is given by the following result. 

\begin{prop}\label{CSMcone}
Let $i: X\hookrightarrow \P^m$ be a subvariety, and let 
$j: C(X)\hookrightarrow \P^{m+1}$ be the cone over $X$. 
Let $H$ denote the hyperplane class and let
\[
i_* c(X)=f(H)\cap [\P^m]_*
\]
be the CSM class of $X$ expressed in the Chow group (homology) of the
ambient $\P^m$. Then the CSM class of the cone (in the ambient
$\P^{m+1}$) is given by
\begin{equation}\label{coneCSM}
j_* c(C(X))=(1+H) f(H)\cap [\P^{m+1}]_* + [\P^0]_* .
\end{equation}
\end{prop}

\proof Let $j_* c(C(X))=g(H)\cap [\P^{m+1}]_*$.
Notice that $X$ may be viewed as a general hyperplane section of
$C(X)$. Then, by Claim~1 of \cite{Alu6} we have
\[
f(H)\cap [\P^m]_* =i_* c(X)=H\cdot (1+H)^{-1} \cap j_* c(C(X))
=H(1+H)^{-1} g(H) \cap [\P^{m+1}]_* .
\]
This implies
\begin{equation}\label{fgHcsm}
(1+H) f(H) \cap [\P^m]_* = g(H)\cap [\P^m]_* .
\end{equation}
This determines all the coefficients in $g(H)$ with 
the exception of the coefficient of $H^{m+1}$. The latter 
equals the Euler characteristic of $C(X)$, hence by Lemma
\ref{EulCharCone} this is $\chi(C(X))=\chi(X)+1$. Thus, we have 
\[
\text{coefficient of $H^{m+1}$ in $g(H)$} = 1+
\text{coefficient of $H^m$ in $f(H)$}\quad.
\]
Together with \eqref{fgHcsm}, this implies \eqref{coneCSM}. 
\endproof

This result applies to some of the operations on graphs described
above. Here, as in the rest of the paper, we suppress the explicit
pushfoward notation $i_*$ and $j_*$ in writing CSM classes in the Chow
group or homology of the ambient projective space.

\begin{cor}\label{subdivide}
Let $\Gamma$ be a graph with $n$ edges, and let $\hat\Gamma$ be the 
graph obtained by subdividing an edge
or by attaching a tree consisting of a single edge to one of the vertices. 
If the CSM class of the hypersurface $X_\Gamma$ is of the form 
$c(X_\Gamma)=f(H) \cap [\P^{n-1}]_*$, with $f$ a polynomial of
$\deg(f)\leq n-1$ in the hyperplane class, then the class of
$X_{\hat\Gamma}$ is given by
\[
c(X_{\hat\Gamma})=((1+H)f(H)+H^n)\cap [\P^n]_* \quad.
\]
\end{cor}

\proof The result follows immediately from Proposition \ref{CSMcone},
since in the first case the graph polynomial $\Psi_{\hat\Gamma}$ depends
on a pair of variables $t_i,t_j$ only through their sum $t_i+t_j$,
hence $X_{\hat\Gamma}$ is a cone over $X_{\Gamma}$ inside $\P^n$. In
the second case the graph polynomial $\Psi_{\hat\Gamma}$ is independent
of the variable of the additional edge and the result follows by the
same argument, since $X_{\hat\Gamma}$ is then also a cone over $X_\Gamma$.
\endproof

The case of attaching an arbitrary tree to a vertex of the graph is
obtained by iterating the second case of Corollary \ref{subdivide}. 

\smallskip

There are further cases of simple operations on a graph which can be analyzed as
an easy consequence of the formulae for cones:
\begin{itemize}
\item Adjoining a loop made of a single edge connecting a vertex to itself.
\item Doubling a disconnecting edge in a non-1PI graph. 
\end{itemize}
In these cases the resulting graph hypersurface is obtained by first
taking a cone over the original hypersurface in one extra dimension
and then taking the union with a transversal hyperplane, respectively
given by the vanishing of the coordinate corresponding to the loop edge or by
the vanishing of the sum $t_i+t_j$ coming from the pair of parallel edges. We then have the
following result.

\begin{cor}\label{addaloop}
Let $\Gamma$ be a graph with $n$ edges, and let $\Gamma'$ be the 
graph obtained by attaching a looping edge to a vertex of $\Gamma$,
or let $\Gamma$ be a non-1PI graph and let $\Gamma'$ be obtained
from $\Gamma$ by doubling a disconnecting edge.
Suppose that the CSM class of $X_\Gamma$ is of the form
$c(X_\Gamma)=f(H)\cap [\P^{n-1}]_*$, for a polynomial of degree~$\le
n-1$ in the hyperplane class. Then the CSM class of $X_{\Gamma'}$ is
given by 
\begin{equation}\label{CSMloop}
c(X_{\Gamma'})=(f(H)+((1+H)^n-H^n) H + H^n)\cap [\P^n]_*\quad.
\end{equation}
\end{cor}

\proof In this case, $X_{\Gamma'}$ is obtained by taking the union of the cone
over $X_\Gamma$ with a general hyperplane $\cL$. Since the intersection 
of a general hyperplane and $X_{\Gamma'}$ is nothing but $X_\Gamma$ itself,
the inclusion-exclusion property for CSM classes discussed in \S
\ref{CSMvsGrSec} gives 
\begin{equation}\label{eqCSMloop}
\begin{array}{rl}
c(X_{\Gamma'})= & c(X_{\hat\Gamma})+ c(\cL)- c(X_\Gamma)\\[2mm]
=& ((1+H)f(H)+H^n + ((1+H)^n-H^n)H -H f(H) )\cap [\P^n]_*
\end{array}
\end{equation}
as claimed.
\endproof

A general remark that may be worth making is the consequence of these
results for the positivity question of Conjecture \ref{posconj}.

\begin{cor}\label{posCones}
If $X_\Gamma$ has positive CSM class, and $\underline\Gamma$
is obtained from $\Gamma$ by any of the operations listed above
(subdividing edges, doubling disconnecting edges, attaching trees 
and single-edge loops), then $X_{\underline\Gamma}$ also has 
positive CSM class.
\end{cor}

The case of joining two
graphs by a pair of edges
operation mentioned above generalizes the one of doubling a
disconnecting edge but is more difficult to deal with explicitly. 
Given a pair of 1PI graphs $\Gamma_1$ and $\Gamma_2$ and
two additional edges joining them as in Figure \ref{FigJoin}, the
graph polynomial becomes of the form
\begin{equation}\label{PsiGamma12egdes}
\Psi_\Gamma(t)=(t_1+t_2)\Psi_{\Gamma_1}(t_3,\ldots,t_{n_1+2})\Psi_{\Gamma_2}(t_{n_1+3},
\ldots, t_{n_1+n_2+2}) +\Psi_{\Gamma_1,\Gamma_2}(t_3,\ldots,t_{n_1+n_2+2}),
\end{equation}
where $n_1=\# E(\Gamma_1)$ and $n_2=\# E(\Gamma_2)$. Here the first
term corresponds to the spanning trees of $\Gamma$ that contain either
the edge $t_1$ or $t_2$, while the second term comes from the spanning
trees that contain both of the additional edges $t_1$ and $t_2$. The
resulting hypersurface $X_\Gamma\subset \P^{n_1+n_2+1}$ is once again
a cone since it depends on the variables $t_1$ and $t_2$ only through
their sum. However, in this case one does not have a direct control on
the form of the CSM class in terms of those of $X_{\Gamma_1}\subset
\P^{n_1-1}$ and $X_{\Gamma_2}\subset \P^{n_2-1}$. Thus, we do not
treat this case here. 

\begin{center}
\begin{figure}
\includegraphics[scale=0.6]{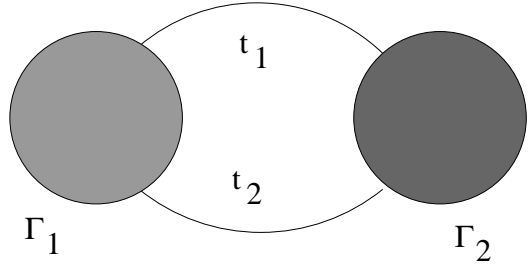}
\caption{Joining two 1PI graphs by a pair of edges. \label{FigJoin}}
\end{figure}
\end{center}

We can proceed similarly to give the relation between classes in the
Grothendieck ring. This is in fact easier than the case of CSM
classes.

\begin{prop}\label{GrCones}
With notation as in Corollaries~\ref{subdivide} and~\ref{addaloop}, we
have
\begin{equation}\label{eqGrCone}
\begin{array}{rl}
[X_{\hat\Gamma}]= & (1+\bT)\cdot [X_\Gamma]+[\P^0] \\[2mm]
[X_{\Gamma'}] =& \bT\cdot [X_\Gamma] + [\P^{n-1}]+[\P^0] .
\end{array}
\end{equation}
\end{prop}

\proof The class of a cone in the Grothendiek ring is just
$$ [C(X)]=[(X\times \A^1) \cup \{ pt \}]=[X][\A^1] + [\A^0]. $$
The result then follows immediately.
\endproof

As we have already noticed in our computation of the class in the
Grothendieck ring and of the CSM class in the special case of the
banana graphs, the formulae look nicer when written in terms of the
hypersurface complement, rather than of the hypersurface itself. The
same happens here. When we reformulate the above in terms of the
complements of the hypersurfaces in projective space we find the
following immediate consequence of additivity and of the formulae
obtained previously.

\begin{cor}\label{ConeComplement}
With notation as in Corollaries~\ref{subdivide} and~\ref{addaloop},
assume that $c(\P^{n-1}\smallsetminus X_\Gamma)=
g(H)\cap [\P^{n-1}]_*$. Then
\begin{equation}\label{CSMconecompl}
\begin{array}{rl}
c(\P^n\smallsetminus X_{\hat\Gamma}) = & (1+H)g(H)\cap [\P^n]_* \\[2mm]
c(\P^n\smallsetminus X_{\Gamma'}) = & g(H)\cap [\P^n]_*
\end{array}
\end{equation}
Similarly, the classes in the Grothendieck group satisfy
\begin{equation}\label{Grconecompl}
\begin{array}{rl}
[\P^n\smallsetminus X_{\hat\Gamma}] =& (1+\bT)\cdot [\P^{n-1}
\smallsetminus X_\Gamma] \\[2mm]
[\P^n\smallsetminus X_{\Gamma'}] =& \bT\cdot [\P^{n-1}
\smallsetminus X_\Gamma] .
\end{array}
\end{equation}
\end{cor}

\begin{center}
\begin{figure}
\includegraphics[scale=0.9]{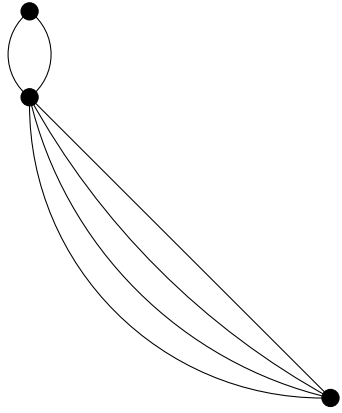}
\caption{Graph obtained from $\Gamma_4$ by adding loops and
subdividing edges. \label{Fig4Ban2}}
\end{figure}
\end{center}

We give two explicit examples obtained from the banana graphs by
applying the operations discussed above.

\begin{ex}\label{4bananaEx}{\rm
Attach a looping edge to the banana graph $\Gamma_4$ of
Figure \ref{BananaFig} and then subdivide the new edge. 
This gives the graph in Figure \ref{Fig4Ban2}. 
The CSM class of the hypersurface $X_{\Gamma_4}$ is 
$$ (3H+3H^2+5H^3)\cap [\P^3]_* $$
by Theorem \ref{CSMXGammaThm}. According to Corollary~\ref{addaloop},
adding a loop gives a graph whose hypersurface has CSM class
\[
((3H+3H^2+5H^3)+((1+H)^4-H^4)H+H^4)\cap [\P^4]_*
=(4H+7H^2+11H^3+5H^4)\cap [\P^4]_* .
\]
Subdividing the new edge (or any other edge) produces a hypersurface 
whose CSM class is given by
\[
((1+H)(4H+7H^2+11H^3+5H^4)+H^5)\cap [\P^5]_*
=(4H+11H^2+18H^3+16H^4+6H^5)\cap [\P^5]_* .
\]
This is the CSM class corresponding to the graph in the picture.
In the Grothendieck group of varieties, the class of the complement
of the hypersurface $X_{\Gamma_4}$ is 
\[
\bT^3+3\bT^2+\bT-1 .
\]
It follows immediately that the class of the complement of the 
hypersurface of the graph in Figure \ref{Fig4Ban2} is then
\[
\bT(\bT+1)(\bT^3+3\bT^2+\bT-1) = 
\bT^5+4\bT^4+4\bT^3-\bT .
\] }
\end{ex}

\begin{ex}\label{bananasplit}{\rm 
Splitting one edge in a banana graph (see Figure \ref{Fig4Ban})
produces a particularly simple class in the Grothendieck group for the
complement of the corresponding hypersurface. The class for the
`banana split' graph is
\[
\bT^n+n\bT^{n-1}+n\bT^{n-2}-(-1)^n .
\] }
\end{ex}

\begin{center}
\begin{figure}
\includegraphics[scale=0.9]{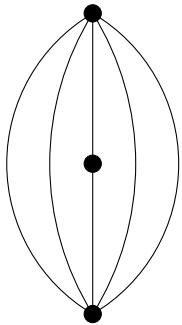}
\caption{Banana split graph. \label{Fig4Ban}}
\end{figure}
\end{center}

\section{Banana graphs in Noncommutative QFT}\label{NCQFTsect}

Recently there has been growing interest in investigating the
renormalization properties and the perturbative theory for certain
quantum field theories on noncommutative spacetimes. These arise, 
for instance, as effective limits of string theory \cite{CDS}, 
\cite{SW}. In particular, in dimension $D=4$, when the underlying
$\R^4$ is made noncommutative by
deformation to $\R^4_\theta$ with the Moyal product, it is known that
the $\phi^4$ theory behaves in a very interesting way. In particular,
the Grosse--Wulkenhaar model was proved to be renormalizable to all
orders in perturbation theory (for an overview see \cite{GroWu}). 
We do not recall here the main aspects
of noncommutative field theory, as they are beyond the main purpose of
this paper, but we mention the fact, which is very relevant to us,
that a parametric representation for the Feynman integrals exists 
also in the noncommutative setting (\cf \cite{GuRi}, \cite{Tana}).
When the underlying spacetime becomes noncommutative, the usual
Feynman graphs are replaced by ribbon graphs, which account for the
fact that, in this case, in the Feynman rules the contribution of each
vertex depends on the cyclic ordering of the edges, \cf \cite{GroWu}.
For example, in the ordinary commutative case, among the banana graphs
$\Gamma_n$ we consider in this paper the only ones that can appear as 
Feynman graphs of the $\phi^4$ theory are the one loop case (with two
external edges at each vertex), the two loop case (with one external
edge at each vertex) and the three loop case as a vacuum bubble. 
Excluding the vacuum bubble
because of the presence of the polynomial $P_\Gamma(t)$, we see that
the effect of making the underlying spacetime noncommutative turns the
remaining two graphs into the graphs of Figure
\ref{NCBananaFig}. Notice how the two loop ribbon graph now has two
distinct versions, only one of which is a planar graph. The usual
Kirchhoff polynomial $\Psi_\Gamma(t)$ of the Feynman graph, as well as
the polynomial $P_\Gamma(t,p)$, are replaced by new polynomials
involving pairs of spanning trees, one in the graph itself and one in
another associated graph which is a dual graph in the planar case and
that is obtained from an embedding of the ribbon graph on a Riemann
surface in the more general case. Unlike the commutative case, these
polynomials are {\em no longer homogeneous}, hence the corresponding
graph hypersurface only makes sense as an affine hypersurface. The
relation of the hypersurface for the noncommutative case and the
one of the original commutative case (also viewed as an affine
hypersurface) is given by the following statement.

\begin{prop}\label{NCgraphhypers}
Let $\tilde\Gamma$ be a ribbon graph in the
noncommutative $\phi^4$-theory that corresponds in the ordinary
$\phi^4$-theory to a graph $\Gamma$ with $n$ internal edges.
Then instead of a single graph hypersurface $X_\Gamma$ one has a
one-parameter family of affine hypersurfaces
$X_{\tilde\Gamma,s}\subset \A^n$, where the parameter $s\in \R_+$
depends upon the deformation parameter $\theta$ of the noncommutative
$\R^4_\theta$ and on the parameter $\Omega$ of the harmonic oscillator
term in the Grosse--Wulkenhaar model. The hypersurface corresponding
to the value $s = 0$ has a singularity at the origin $0\in \A^n$ whose
tangent cone is the (affine) graph hypersurface $X_\Gamma$.
\end{prop}

\proof This follows directly from the relation between the graph
polynomial for the ribbon graph $\tilde\Gamma$ given in \cite{GuRi}
and the Kirchhoff polynomial $\Psi_\Gamma$. It suffices to see that
(a constant multiple of) the Kirchhoff polynomial is contained in the
polynomial for $\tilde\Gamma$ for all values of the parameter $s$, and
that it gives the part of lowest order in the variables $t_i$ when
$s=0$. 
\endproof 

In the specific examples of the banana graphs $\tilde\Gamma_2$ and
$\tilde\Gamma_3$ of Figure \ref{NCBananaFig}, the polynomials have
been computed explicitly in \cite{GuRi} and they are of the form
\begin{equation}\label{NCpolynB2}
\Psi_{\tilde\Gamma_2} = (1+4s^2) (t_1+t_2+ t_1^2t_2+t_1t_2^2),
\end{equation}
where the parameter $s=(4\theta\Omega)^{-1}$ is a function of the
deformation parameter $\theta\in \R$ of the Moyal plane and of the
parameter $\Omega$ in the harmonic oscillator term in the
Grosse--Wulkenhaar action functional (see \cite{GroWu}). One can see
the polynomial $\Psi_{\Gamma_2}(t)=t_1+t_2$ appearing as lowest order
term. Similarly for the two graphs $\tilde\Gamma_3$ that correspond to
the banana graph $\Gamma_3$ one has (\cite{GuRi})
\begin{equation}\label{NCpolynB3plan}
\begin{array}{rl}
\Psi_{\tilde\Gamma_3}(t)= & (1+8s^2+16s^4) (t_1t_2+t_2t_3+t_1t_3+
t_1^2t_2t_3 + t_1 t_2^2 t_3 + t_1 t_2 t_3^2) \\[2mm] + & 16 s^2 (t_2^2
+t_1^2t_3^2) 
\end{array}
\end{equation}
for the planar case, while for the non-planar case one has
\begin{equation}\label{NCpolynB3noplan}
\begin{array}{rl}
\Psi_{\tilde\Gamma_3}(t)= & (1+8s^2+16s^4) (t_1t_2+t_2t_3+t_1t_3+
t_1^2t_2t_3 + t_1 t_2^2 t_3 + t_1 t_2 t_3^2)  \\[2mm] + & 4s^2 (t_2^2
+t_1^2t_3^2 + t_1^2 +t_2^2 t_3^2 + t_3^2 + t_1^2t_2^2 + 1 +
t_1^2t_2^2t_3^2). 
\end{array}
\end{equation}
In both cases, one readily recognizes the polynomial
$\Psi_{\Gamma_3}(t)= t_1t_2+t_2t_3+t_1t_3$ as the lowest order part at
$s=0$. Notice how, when $s\neq 0$ one finds other terms of order less
than or equal to that of the polynomial $\Psi_{\Gamma_3}(t)$, such as 
$t_2^2$ in \eqref{NCpolynB3plan} and $1+ t_1^2+t_2^2+t_3^2$ in
\eqref{NCpolynB3noplan}. 
Notice also how, at the limit value $s= 0$ of the parameter, the 
two polynomials for the two different ribbon graphs corresponding to
the third banana graph $\Gamma_3$ agree.

For each value of the parameter $s=(4\theta\Omega)^{-1}$ one obtains
in this way an affine hypersurface, which is a curve in $\A^2$ or a
surface in $\A^3$, and that has the corresponding affine $X_{\Gamma_n}$ as
tangent cone at the origin in the case $s=0$. 
The latter is a line in the $n=2$ case and a cone on a nonsingular
conic in the case $n=3$. 

\begin{center}
\begin{figure}
\includegraphics[scale=0.4]{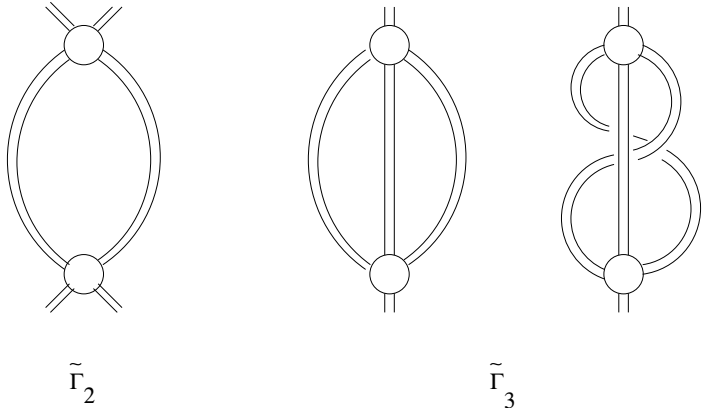}
\caption{Banana graphs in noncommutative $\phi^4$-theory
\label{NCBananaFig}}
\end{figure}
\end{center}

As a further example of why it is useful to compute invariants such as
the CSM classes for the graph hypersurfaces, we show that the CSM
class of the hypersurface defined by the polynomial
\eqref{NCpolynB3plan} detects the special values of the deformation
parameter $s=(4\theta\Omega)^{-1}$ where the hypersurface
$X_{\tilde\Gamma_3}$ becomes more singular and gives a quantitative
estimate of the amount by which the singularities change. 

The CSM class is naturally defined for projective varieties. In the
case of an affine hypersurface defined by a non-homogeneous equation, 
one can choose to compactify it in projective space by adding an extra
variable and making the equation homogeneous and then computing the
CSM class of the corresponding projective hypersurface. However, in
doing so one should be aware of the fact that the CSM class
of an affine variety, defined by choosing an embedding in a larger
compact ambient variety, depends on the choice of the embedding. An
intrinsic definition of CSM classes for non-compact varieties which
does not depend on the embedding was given in \cite{Alu2},
\cite{Alu06a}. However, for our purposes here it suffices to take the
simpler viewpoint of making the equation homogeneous and then
computing CSM classes. If we adopt this procedure, then by numerical
calculations performed with the Macaulay2 program of \cite{Alu3} 
we obtain the following result.

\begin{prop}\label{CSMbananaNC}
Let $X_{\tilde\Gamma_3}\subset \P^3$ denote the affine surface defined
by the equation \eqref{NCpolynB3plan} and let $\bar
X_{\tilde\Gamma_3}\subset \P^3$ be the hypersurface obtained by making
the equation \eqref{NCpolynB3plan} homogeneous. For general values of the
parameter $s=(4\theta\Omega)^{-1}$ the CSM class is given by
\begin{equation}\label{CSMgeneric}
c(\bar X_{\tilde\Gamma_3})= 14H^3 + 4H.
\end{equation}
For the special value $s=1/2$ of the parameter, the CSM class becomes
of the form
\begin{equation}\label{CSMs12}
c(\bar X_{\tilde\Gamma_3})|_{s=1/2}=  5H^3  + 5H^2  + 4H,
\end{equation}
while in the limit $s\to 0$ one has
\begin{equation}\label{CSMs0}
c(\bar X_{\tilde\Gamma_3})|_{s=0}=  11H^3 + 4H.
\end{equation}
\end{prop}

It is also interesting to notice that, when we consider the second
equation \eqref{NCpolynB3noplan} for the non-planar ribbon graph
associated to the third banana graph $\Gamma_3$, we see an example
where the graph hypersurfaces of the non-planar graphs of
noncommutative field theory no longer satisfy the positivity property of
Conjecture \ref{posconj} that appears to hold for the graph
hypersurfaces of the commutative field theories. In fact, 
as in the case of the equation for the planar graph
\eqref{NCpolynB3plan}, we now find the following result.

\begin{prop}\label{CSMbananaNCnonplan}
Let $X_{\tilde\Gamma_3}\subset \P^3$ denote the affine surface defined
by the equation \eqref{NCpolynB3noplan} and let $\bar
X_{\tilde\Gamma_3}\subset \P^3$ be the hypersurface obtained by making
the equation \eqref{NCpolynB3noplan} homogeneous. For general values of the
parameter $s=(4\theta\Omega)^{-1}$ the CSM class is given by
\begin{equation}\label{CSMgenericNP}
c(\bar X_{\tilde\Gamma_3})=33H^3  - 9H^2  + 6H.
\end{equation}
The special case $s=1/2$ is given by
\begin{equation}\label{CSMgenericNP12}
c(\bar X_{\tilde\Gamma_3})|_{s=1/2}= 9H^3  - 3H^2  + 6H
\end{equation}
\end{prop}

Notice that, in the case of ordinary Feynman graphs of commutative 
scalar field theories, all the examples where the CSM classes of 
the corresponding hypersurfaces have been computed explicitly 
(either theoretically or numerically) are planar graphs. 
Although it seems unlikely that planarity will play a role in 
the conjectured positivity of the coefficients of the CSM classes 
in the ordinary case, the example above showing that CSM classes 
of graph hypersurfaces of non-planar ribbon graphs in noncommutative
field theories can have negative 
coefficients makes it more interesting to check the case of non-planar
graphs in the ordinary case as well. It is well known that, for an
ordinary graph to be non-planar, it has to contain a copy of either
the complete graph $K_5$ on five vertices or the complete bipartite
graph $K_{3,3}$ on six vertices. Either one of these graphs corresponds
to a graph polynomial that is currently beyond the reach of the
available Macaulay2 routine and a theoretical argument that provides a
more direct approach to the computation of the corresponding CSM class
does not seem to be easily available. It remains an interesting
question to compute these CSM classes, especially in view of the fact
that the original computations of \cite{BroKr} of Feynman integrals of
graphs appear to indicate that the non-planarity of the graph is
somehow related to the presence of more interesting periods (\eg
multiple as opposed to simple zeta values). It would be interesting to
see whether it also has an effect on invariants such as the CSM class.



\begin{thebibliography}{999999}

\bibitem{Alu06a} P.~Aluffi, {\em  Classes de Chern des vari\'et\'es 
singuli\`eres, revisit\'ees}. C. R. Math. Acad. Sci. Paris  342
(2006),  no. 6, 405--410. 

\bibitem{Alu2} P.~Aluffi, {\em Limits of Chow groups, and a new
construction of Chern-Schwartz-MacPherson classes}.  Pure
Appl. Math. Q.  2  (2006),  no. 4, 915--941. 

\bibitem{Alu1} P.~Aluffi, {\em Modification systems and integration in
their Chow groups}.  Selecta Math. (N.S.)  11  (2005),  no. 2,
155--202.  

\bibitem{Alu3} P.~Aluffi, {\em Computing characteristic classes of
projective schemes}.  J. Symbolic Comput.  35  (2003),  no. 1, 3--19. 

\bibitem{Alu5} P.~Aluffi, 
{\em Chern classes for singular hypersurfaces}.
Trans. Amer. Math. Soc. 351 (1999), no. 10, 3989--4026. 

\bibitem{Alu6} P.~Aluffi, 
{\em MacPherson's and Fulton's Chern classes of hypersurfaces}, 
Internat. Math. Res. Notices, Vol.11 (1994) 455--465.

\bibitem{AluMih} P.~Aluffi, L.C.~Mihalcea, {\em Chern classes of
Schubert cells and varieties}, arXiv:math/0607752. To appear in
Journal of Algebraic Geometry.

\bibitem{BeBro} P.~Belkale, P.~Brosnan, {\em Matroids, motives, and
a conjecture of Kontsevich}, Duke Math. Journal, Vol.116 (2003)
147--188.

\bibitem{BjDr} J.~Bjorken, S.~Drell, {\em Relativistic Quantum
Mechanics}, McGraw-Hill, 1964, and {\em Relativistic Quantum Fields}, 
McGraw-Hill, 1965.

\bibitem{Blo} S.~Bloch, {\em Motives associated to graphs},
Japan J. Math., Vol.2 (2007) 165--196.

\bibitem{BEK} S.~Bloch, E.~Esnault, D.~Kreimer, {\em On motives
associated to graph polynomials}, Commun. Math. Phys., Vol.267
(2006) 181--225.

\bibitem{BloKr} S.~Bloch, D.~Kreimer, {\em Mixed Hodge structures and
renormalization in physics}, arXiv:0804.4399.

\bibitem{BLSS} J.-P. Brasselet, D.~Lehmann, J.~Seade, T.~Suwa,
{\em Milnor classes of local complete intersections}, 
Trans. Amer. Math. Soc., Vol.354(2002) N.4, 1351--1371.

\bibitem{BraSchwa} J.P.~Brasselet, M.H.~Schwartz, {\em Sur les classes
de Chern d'un ensemble analytique complexe}, in ``The Euler-Poincar\'e
characteristic'', Ast\'erisque, Vol.83, pp.93--147, 
Soc. Math. France, 1981.

\bibitem{BroKr} D.~Broadhurst, D.~Kreimer, {\em Association of
multiple zeta values with positive knots via Feynman diagrams up to
9 loops}, Phys. Lett. B, Vol.393 (1997) 403--412.

\bibitem{CDS} A.Connes, M.Douglas, A.Schwarz, {\em Noncommutative
geometry and matrix theory: compactification on tori}. JHEP 9802
(1998) 3--43.

\bibitem{CK} A.~Connes, D.~Kreimer, {\em Renormalization in
quantum field theory and the Riemann--Hilbert problem I. The Hopf
algebra structure of graphs and the main theorem}, Comm. Math.
Phys., Vol.210 (2000) 249--273.

\bibitem{Fulton} W.~Fulton, {\em Intersection theory}.
Ergebnisse der Mathematik und ihrer Grenzgebiete (3) 
2. Springer-Verlag, 1984. xi+470 pp.

\bibitem{GilSou} H.~Gillet, C.Soul\'e, {\em Descent, motives and $K$-theory}.
J. Reine Angew. Math. 478 (1996), 127--176. 

\bibitem{Macaulay2} D.R.~Grayson, M.E.~Stillman, {\em 
          Macaulay 2, a software system for research
                   in algebraic geometry},
available at {\tt{http://www.math.uiuc.edu/Macaulay2/}}

\bibitem{GroWu} H.~Grosse, R.~Wulkenhaar, {\em Renormalization of
noncommutative quantum field theory}, in ``An invitation to
Noncommutative Geometry'' pp.129--168, World Scientific, 2008.

\bibitem{GuRi} R.~Gurau, V.~Rivasseau, {\em  Parametric Representation
of Noncommutative Field Theory}, Commun. Math. Phys. Vol. 272 (2007)
N.3, 811--835

\bibitem{ItZu} C.~Itzykson, J.B.~Zuber, {\em Quantum Field Theory},
Dover Publications, 2006.

\bibitem{Kwiec} M.~Kwieci{\'n}ski, {\em
Formule du produit pour les classes caract\'eristiques de
  Chern-Schwartz-MacPherson et homologie d'intersection},
C. R. Acad. Sci. Paris S\'er. I Math., 314(1992) N.8, 625--628.

\bibitem{MacPher} R.D.~MacPherson, {\em Chern classes for singular
algebraic varieties}.  Ann. of Math. (2) 100 (1974), 423--432.  

\bibitem{Mar} M.~Marcolli, {\em Motivic renormalization and
singularities}, arXiv:0804.4824. 

\bibitem{Naka} N.~Nakanishi {\em Graph Theory and Feynman Integrals}. 
Gordon and Breach, 1971.

\bibitem{Parus} A.~Parusi{\'n}ski, {\em A generalization of the Milnor
number}, Math. Ann. Vol.281 (1988) N.2,  247--254.

\bibitem{MHSchwartz1} M.H.~Schwartz, {\em Classes caract\'eristiques
d\'efinies par une stratification d'une vari\'et\'e analytique complexe. I.}
C. R. Acad. Sci. Paris Vol.260 (1965) 3262--3264. 

\bibitem{MHSchwartz2} M.H.~Schwartz, {\em Classes caract\'eristiques 
d\'efinies par une stratification d'une vari\'et\'e analytique complexe},
C. R. Acad. Sci. Paris, Vol.260 (1965) 3535--3537.

\bibitem{SW} N.~Seiberg, E.~Witten, {\em String theory and
noncommutative geometry}, JHEP 9909 (1999) 32--131.

\bibitem{Tana} A.~Tanasa, {\em Overview of the parametric
representation of renormalizable non-commutative field theory},
arXiv:0709.2270. 

\end{thebibliography}
\end{document}